\documentclass[12pt,a4,english,leqno]{article}
\usepackage[latin9]{inputenc}

\usepackage{geometry}
\usepackage{verbatim}
\usepackage{amsmath}
\usepackage{amsfonts}
\usepackage{amsthm}
\usepackage{amssymb}
\usepackage{mathtools}
\usepackage{graphicx}
\usepackage{longtable}
\usepackage{setspace}
\usepackage{comment}
\usepackage{xcolor}
\usepackage{upgreek}

\usepackage{lmodern}
\usepackage{avant}

\usepackage[round,sort,authoryear]{natbib}
\usepackage[colorlinks=true,linkcolor=blue, citecolor=blue, urlcolor = blue]{hyperref}

\usepackage{titlesec}
\usepackage{babel}
\usepackage[toc,page]{appendix}
\usepackage{booktabs,caption}

\usepackage[labelfont=bf,textfont=md]{caption}
\usepackage[labelsep=space]{caption}

\usepackage{lmodern}
\usepackage[T1]{fontenc}

\DeclarePairedDelimiter{\floor}{\lfloor}{\rfloor}

\newcommand\norm[1]{\left\lVert#1\right\rVert}

\usepackage{bibentry}
\nobibliography*

\usepackage{fancyhdr}
\numberwithin{equation}{section}

\titleformat*{\section}{\large \bfseries}
\titleformat*{\subsection}{\normalsize \bfseries}
\titleformat*{\subsubsection}{\small \bfseries}

\providecommand{\U}[1]{\protect\rule{.1in}{.1in}}
\newtheorem{theorem}{Theorem}[section]

\newtheorem{corollary}{Corollary}[section]

\newtheorem{example}{Example}

\newtheorem{lemma}{Lemma}[section]

\newtheorem{remark}{Remark}[section]

\newtheorem{assumption}{Assumption}[section]

\usepackage{mathtools}

\numberwithin{equation}{section}

\AtBeginDocument{}

\newif\ifshow 
\showfalse 

\ifshow
  \includecomment{wrap}
\else
  \excludecomment{wrap} 
  
\fi

\begin{document}
\pagenumbering{roman}

\title{ {\Large \textbf{Asymptotic Theory for Unit Root Moderate Deviations in \\ \vspace{-0.1\baselineskip} Quantile Autoregressions and Predictive Regressions}\thanks{\textbf{Article history:} Previous draft titled: "\textit{Asymptotic Theory for Moderate Deviations from the Unit Boundary in Quantile Autoregressive Time Series}", 05 April 2022. This draft: August 2023. I am grateful to Xiaohui Lui, Tassos Magdalinos, Jean-Yves Pitarakis and Jose Olmo for helpful discussions.\\  \\ Lecturer in Economics, University of Exeter Business School, Exeter EX4 4PU, United Kingdom. \textit{Email:} \textcolor{blue}{christiskatsouris@gmail.com}. }} \\
}

\author{\textbf{Christis Katsouris}\\
\textit{University of Southampton and University of Exeter}\\
\\ Working Paper}


\date{\today}

\maketitle

\begin{abstract}
\vspace{-0.71\baselineskip}
We establish the asymptotic theory in quantile autoregression when the model parameter is specified with respect to moderate deviations from the unit boundary such that $\rho_n = \left(  1 + \frac{c}{k_n} \right)$ where $\left( k_n \right)_{ n \in \mathbb{N} }$ is a nonrandom sequence that diverges at a rate slower than the sample size $n$. Then, extending the framework proposed by  \cite{Phillips2007limit}, we consider the limit theory for the near-stationary and the near-explosive cases when the model is estimated with a conditional quantile functional form and model parameters are quantile-dependent. A Bahadur-type representation and limiting distributions based on the M-estimators of the model parameters are derived. We show that the serial correlation coefficient converges in distribution to a ratio of two independent random variables. Monte Carlo simulations illustrate the finite-sample performance of the estimation procedure under investigation.   
\\

\\


\end{abstract}

\newpage 

\setcounter{page}{1}
\pagenumbering{arabic}

\section{Introduction}

Moderate deviation principles from the unit boundary for quantile autoregressions are commonly employed when considering the limit distribution of quantile-dependent parameters under regressors nonstationarity. In particular, the development of asymptotic theory for nonstationary quantile time series models has been pioneered by the studies of \cite{koenker2004unit, koenker2006quantile} as well as \cite{koenker2002inference} who investigate estimation and inference aspects for regression quantile models (see, also \cite{hasan1997robust}). Specifically, studies for quantile autoregressive regressions that consider moderate deviations within a unified framework  allowing to investigate the asymptotic behaviour of estimators with respect to different regimes of stability such as stable, unstable and explosive processes has seen less attention in the literature. Therefore, our main objective is to use the moderate deviation principles in order to derive the limiting distribution of the autoregressive coefficient when considering deviations from the unit boundary under a conditional quantile functional form. The present paper builds on the framework proposed by  \cite{Phillips2007limit} (see, also \cite{giraitis2006uniform} and \cite{huang2014limit}) that corresponds to the linear autoregressive model under nonstationarity as well as the study of \cite{kong2015m} and \cite{wang2022asymptotics} who develop limit theory for moderate deviations in autoregressive models based on M-estimators.

In this line of literature, the two relevant research questions are summarized by \cite{chan2006quantile}: 

\begin{quotation}
"\textit{The study of the unit root AR(1) model has been actively pursued by statisticians and econometricians alike, and a related question that needs to be addressed is what happens to the limiting distribution of the test statistics when the autoregressive parameter $\theta_n$ is close to the unit boundary? Consequently, when the autoregression coefficient is expressed with respect to the local-to-unity parametrization, what kind of approximation should be used for the distribution of the test statistics?}. 
\end{quotation}
The seminal studies of \cite{chan1987asymptotic1} and \cite{phillips1987time, phillips1987towards} tackle exactly these nonstandard statistical problems via their triangular array framework. The framework of the nearly nonstationary AR(1) model allows to establish the limiting distributions of the least squares estimator for $\theta_n$ under the assumption that the conditional variance of the model is finite. Moreover, the properties of the least squares estimator when the true stochastic process is nearly integrated are investigated by  \cite{chan1988parameter, chan1990inference}, \cite{chan1989first}, \cite{knight1987rate},  \cite{rao1978asymptotic}, \cite{lai1982least}, \cite{cox1991maximum}, \cite{larsson1995asymptotic}, \cite{cavaliere2002bounded},  \cite{Buchmann2007asymptotic},   \cite{Phillips2007limit} and \cite{duffy2021estimation} among others. Furthermore, the asymptotic theory for moderate deviations from a unit root in autoregressive models is presented by \cite{fountis1989testing}, \cite{jiang2015moderate} who focus on the aspect of dependent errors in AR(1) models and \cite{yabe2012limiting} who obtain limit results for MA(1) time series models. 

\newpage

On the other hand, the properties of nonstationary autoregressive models for the case of an explosive autoregressive coefficient is also of interest. In particular, \cite{white1958limiting} obtained the limit distribution of an explosive serial correlation coefficient (see, also \cite{mann1943statistical}). Limit theory for moderate deviations on the explosive side of unity (e.g., mildly explosive case) were developed by  \cite{Buchmann2007asymptotic}, \cite{aue2007limit}, \cite{magdalinos2012mildly}, \cite{arvanitis2018mildly}, \cite{oh2018mildly}, \cite{lee2018limit}, \cite{proia2020moderate}, \cite{yu2021inference}, \cite{hirukawa2021asymptotic}, \cite{liu2022mildly}. Recently, \cite{Jian2022} derived the limit theory for the ordinary least squares estimator in the explosive first-order Gaussian autoregressive process using a set of deviation inequalities\footnote{The authors obtain the limit theory of Cram\'er-type moderate deviations for the explosive and mildly explosive autoregressive processes.}. In particular, \cite{aue2007limit} develop the limit theory for the serial correlation coefficient in the mildly explosive case with moderate deviations from the unit boundary. Thus, the model parameter satisfies $\theta_n \to 1 \ \ \text{and} \ \ n \left( \theta_n - 1 \right) \to \infty$, as $n \to \infty$, while for large $n$, $\theta_n > 1$ diverges away from unity but not with the usual convergence rate $\mathcal{O} ( 1 / n)$.

Under the Cram\'er-type moderate deviations framework, there exists positive sequences $\nu_n$ and $\lambda_n$ tending to infinity such that for every $\ell > 0$ as $n \to \infty$ it holds that (e.g., see \cite{Jian2022})
\begin{align}
\underset{ 0 \leq x \leq \ell \lambda_n }{ \mathsf{sup} } \left| \frac{1}{ 1 - F_n( x ) } \mathbb{P} \bigg( \nu_n \left( \hat{\theta}_n - \theta_n \right) \geq x \bigg) - 1  \right| \to 0,
\end{align}
where $F_n (x) $ is the distribution function, satisfying for all $x \in \mathbb{R}$ 
\begin{align}
\mathbb{P} \bigg( \nu_n \left( \hat{\theta}_n - \theta_n \right) \geq x \bigg) - F_n(x) \overset{ p }{ \to } 0, \ \ \text{as} \ n \to \infty.
\end{align}
We focus on the neighbourhood near the unit boundary, which can be approaching unity from below (near-stationary or near-integrated) or approaching unity from above (near-explosive). Due to the form of the nonrandom sequence such that, $k_n \equiv n^{\gamma}$, the convergence rate towards unity is slower than the sample size $n$. In this case the autoregression  coefficient approaches 1 at a rate slower than the usual local alternative as $n$ goes to infinity. When, as $\gamma \to 1$ then $k_n \to n$ which encompasses the conventional local-to-unity parametrization. An alternative form for the convergence rate include the case when $k_n = \sqrt{n}$. To obtain limit results for model parameters based on M-estimation, we employ the Bahadur representation (see, \cite{bahadur1966note}) that provides a mechanism to facilitate the asymptotic theory, which allows to obtain analytic expressions for quantile-dependent estimators by approximating them with linear forms. Although we do not consider how the presence of serial correlation can affect the limiting distributions, various studies in the literature consider simple implementations of autoregression  robust tests as in \cite{jansson2004error} (see, also \cite{vogelsang1998trend}). Similar implementations can be considered within the modelling environment of quantile regression models especially of those with possibly nonstationary autoregressive processes with serial correlated innovation terms. \cite{kiefer2000simple}, (KVB), demonstrate that the properties of Wald-type statistics can be ameliorated if an inconsistent covariance matrix estimator is used and the critical values are adjusted to accommodate the randomness of the matrix employed in the standardization.

\newpage

\subsection{Literature Review}

Regression asymptotics with roots at or near unity are typically carried out by using autoregressive models with fixed coefficients and then testing for the autoregressive parameter being equal to one, as pointed out by \cite{dickey1979distribution} (see, also \cite{dickey1981likelihood}). The idea of developing asymptotics using the local-to-unity parametrization is due to the studies of \cite{Cavanagh1985LUR},  \cite{phillips1987time}, \cite{chan1987asymptotic1}.  \cite{Phillips2007limit} showed that $\left( \hat{\theta}_n - \theta_n \right)$ has a $\sqrt{n k_n}$ rate of convergence and a limit normal distribution when $c < 0$ for $\theta_n = \left( 1 + \frac{c}{n} \right)$,  
\begin{align}
\sqrt{n k_n} \left( \hat{\theta}_n - \theta_n \right) \overset{ d }{ \to } \mathcal{N} \big( 0, - 2c  \big)
\end{align}
Thus, we are interested to establish a martingale central limit theorem for a normalized version of $\sum_{t=1}^n y_{t-1} u_t$ which can give rise to a Gaussian asymptotic distribution for the normalized and centered least squares estimator specifically for the quantile autoregressive model. 

A different stream of literature considers a  representation of the autoregressive model based on the  exponential family with specific canonical parameter. In that case, by expressing the AR(1) model with respect to the canonical parameters of an exponential family one can establish the asymptotic behaviour of related minimal sufficient statistics\footnote{In particular, \cite{jansson2006optimal} using the local-to-unity parametrization of the autoregression  coefficient study the properties of predictive regressions under the assumption of persistence regressors using differential geometry and sufficient statistics arguments to establish the asymptotic theory of estimators and test statistics.} (see, \cite{jansson2006optimal}). The particular literature is developed under the assumption of a stationary autoregression coefficient such that $| \theta | < 1$ for $y_t = \theta y_{t-1} + u_t$. Furthermore, it has been argued that the Efron curvature depends heavily on the AR parameter, especially near the boundary of the parameter space, and increasingly so with increasing sample size\footnote{Using the local-to-unity parametrization allows the analysis of unit roots and explosive processes, which is necessary to link problems in inference for unit roots to the statistical curvature.} (see, \cite{garderen1999exact}). Specifically, this implies that when the true parameter value is explosive rather than being stationary, then the asymptotic theory of estimators and test statistics based on OLS estimation are driven by the distribution of the innovations $u_t$. On the other hand, comparing the cases of Gaussian innovations against non-Gaussian innovations (e.g., heavy tailed errors) and an explosive autoregressive parameter then the asymptotic theory of the OLS estimator in these two cases will not necessarily be identical. 

Various studies in the literature consider limit results for moderate deviations for M-estimators and quantile processes\footnote{Related theoretical aspects can be found in the book of \cite{csorgHo1983quantile} (see, also \cite{csorgo1986weighted}).} in autoregressive models include among others \cite{jurevckova1988moderate}, \cite{knight1998limiting}, \cite{mao2019moderate} as well as \cite{kato2009asymptotics} who develops asymptotics for Lasso quantile-dependent estimators.  Limit theory for moderate deviations from the unit root in the context of quantile autoregressive models are established in the studies of \cite{lucas1995unit}, \cite{abadir2000quantiles}, \cite{ling2004regression}, \cite{koenker2004unit}, \cite{chan2006quantile}, \cite{kong2015m}, \cite{zhou2015quantile}, \cite{wang2022asymptotics} and \cite{fu2022cqr}.

\newpage

Another important aspect worth mentioning, is the fact that  several studies demonstrated that the nuisance parameter of persistence cannot be consistently estimated (see, \cite{phillips2001estimate}, \cite{mikusheva2012one} among others). Similarly, when considering the quantile autoregressive model, the availability of a consistent estimator for the unknown coefficient of persistence, $c$, still remains a challenging issue. The limit theory for an autoregressive model which includes a  intercept and an autoregression  coefficient expressed using the local-to-unity parametrization is studied by \cite{liu2018limit} and \cite{hwang2009asymptotic}, however their framework differs from our setting since we consider quantile-dependent parameters. This article, builds on and contributes to both the quantile autoregression literature as well as to the literature of moderate deviations from the unit boundary. We focus on establishing the limit theory for moderate deviations from the unit boundary for both the near-stationary and near-explosive cases in quantile autoregression and quantile predictive regressions. Our main objective is to develop a unified framework for the asymptotic behaviour of the quantile-dependent estimators across the whole spectrum of nonstationarity regimes, including the general near-integrated case assuming that $\theta_n \to 1$ with a rate slower than $1/n$.

Throughout the paper, we assume that all random elements are defined within a probability space denoted with the triple $\left( \Omega, \mathcal{F}, \mathbb{P} \right)$. All limits are taken as $n \to \infty$, where $n$ is the sample size. Denote with $\mathcal{D} \left( [0,1] \right)$ to be the set of functions on $[0,1]$ that are right continuous and have left limits, equipped with the Skorokhod metric. Then, the symbol $"\Rightarrow"$ is used to denote the weak convergence of the associated probability measures as $n \to \infty$. The symbol $\overset{ \mathcal{D} }{\to}$ and $\overset{ \mathbb{P} }{ \to }$  are employed to denote convergence in distribution and convergence in probability respectively. Moreover,  we denote with $\mathcal{O}_{ \mathbb{P} }(.)$ and $o_{ \mathbb{P} }(.)$ the stochastic order of convergence in probability (\cite{billingsley1968convergence}). 

The rest of the paper is organized as follows. Section \ref{Section2}, introduces the framework of moderate deviation principles from the unit boundary in quantile autoregressive processes. Section \ref{Section3}, demonstrates the main results of the paper, that is, the limit theory for the near-integrated and near-explosive cases for the quantile autoregression. Section \ref{Section4} presents a short Monte Carlo simulation study while Section \ref{Section5} an empirical application. Section \ref{Section6} concludes.

\section{Moderate Deviations from the Unit Boundary Framework}
\label{Section2}

\subsection{Main Terminology, Definitions and Assumptions}

Consider the first order autoregressive process with expressed as below
\begin{align}
x^{*}_t =  \theta_n x^{*}_{t-1} + u_t, \ \ \ \ \  \text{for} \ \ t = 1,...,n, \ \ \ \ 
\end{align}
with a possible sample-size dependent autoregressive root $\theta_n$ and $\left( u_t \right)_{ t \in \mathbb{N} }$ is the innovation sequence\footnote{The innovations $\left( u_t \right)_{ t \in \mathbb{N} }$ is an \textit{i.i.d} sequence of random variables from a common distribution function $F_{u}$ that satisfies regularity conditions for  Lipschitz continuity with  zero mean and finite variance $\sigma_{u}$. }. 

\newpage

Moreover, we allow for the presence of a non-zero intercept by considering the process
\begin{align}
x_t = \mu + x^{*}_t, \ \ \text{with} \ \ x_{0t} = \sum_{j=1}^n \theta_n^{t-j} u_j.
\end{align}
where $x_{0t}$ represents the solution of the autoregressive process with $\mu_x = 0$ and $X_0 = 0$, the initial condition of the recursion. Furthermore, since we are interested about the asymptotic stability of the autoregressive process near the unit boundary we allows a local-to-unity parametrization as 
\begin{align}
\theta_{n}  =  \left( 1 + \frac{c}{ k_n } \right), \ \text{with} \ \ c \in \mathbb{R}, \ k_n = n^{\gamma} \ \ \text{and} \ \gamma = 0, \ \gamma \in (0,1) \ \text{or} \ \gamma = 1. 
\end{align}
The convergence rate of the autoregression coefficient $\theta_{n}$ is such that $k_n := n^{\gamma}$ where the exponent rate is defined such that $\gamma \in (0,1)$. Specifically, the given convergence rate implies that the $\left( k_n \right)_{ n \in \mathbb{N} }$ sequence increases to infinity at a slower rate than the sample size such that $k_n = o(n)$ as $n \to \infty$.   

\medskip

\begin{assumption}[\textbf{Persistence}]
Consider the autoregressive process $X_t = \theta_n X_{t-1} + \varepsilon_t$ with $\left( \varepsilon_t \right)_{ t \in \mathbb{N}}$. Then, consider the following limit to determine the degree of persistence for the processes 
\begin{align}
\zeta_n := \underset{ n \to + \infty  }{ \mathsf{lim} } n \left( \theta_n - 1 \right) \to \zeta
\end{align}

\begin{itemize}
    
    \item[\textbf{P.1}]  \textbf{\textit{nearly stable}} processes: if $\left( \theta_n \right)_{ n \in \mathbb{N} }$ is such that $\zeta = - \infty$ and it holds that $\theta_n \to | \theta | < 1$.
    
    \item[\textbf{P.2}]  \textbf{\textit{nearly unstable}} processes: if $\left( \theta_n \right)_{ n \in \mathbb{N} }$ is such that $\zeta \equiv c \in \mathbb{R}$ and it holds that $\theta_n \to \theta = 1$.
    
    \item[\textbf{P.3}]  \textbf{\textit{nearly explosive}} processes: if $\left( \theta_n \right)_{ n \in \mathbb{N} }$ is such that $\zeta = + \infty$ and it holds that $\theta_n \to | \theta | > 1$.
    
\end{itemize}

\medskip

\end{assumption}

\begin{assumption}
\label{Assumption0}
Denote with $x_{0n} = 0$, almost surely, for all $n \in \mathbb{N}$. Then, the innovations $\left( \varepsilon_t \right)_{ t \in \mathbb{N} }$ forms a sequence of martingale differences such that as $n \to \infty$  
\begin{align}
\frac{1}{n} \sum_{ t = 1}^n \mathbb{E} \left[ \varepsilon_t^2 \big| \mathcal{F}_{t-1} \right] 
&= 1 + o_p(1),
\\
\frac{1}{n} \sum_{ t = 1}^n \mathbb{E} \left[ \varepsilon_t^2  \mathbf{1} \left\{ | \varepsilon_t \big| > n^{1/2} m \right\} \big| \mathcal{F}_{t-1} \right] 
&= o_p(1), \ \ \text{for all} \ m > 0, 
\end{align}  
where $\mathcal{F}_t := \sigma \big( \varepsilon_s: 0 \leq s \leq t \big)$. 
\end{assumption}

\newpage

\begin{figure}[h!]
\caption[Figure showing ]

\begin{center}
\includegraphics[width=0.35\textwidth]{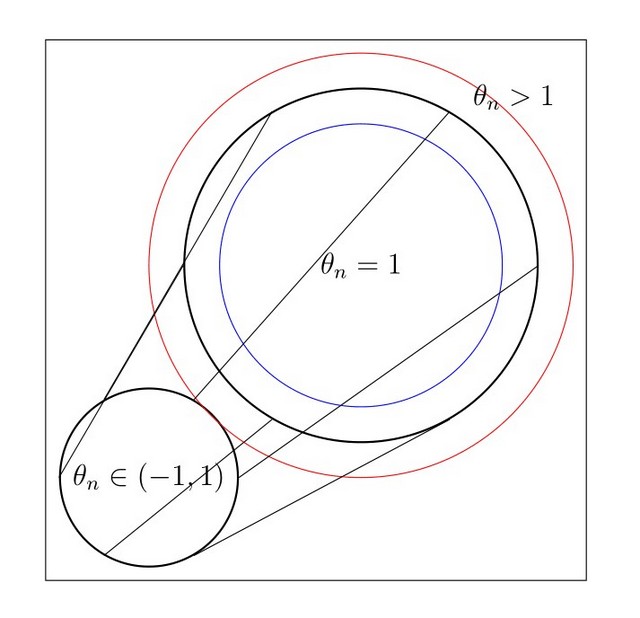}
\end{center}

\end{figure}

\begin{remark}
Assumption \ref{Assumption0} provides moment conditions and a uniform integrability condition which induces a restriction on the tail behaviour of the underline distribution of innovations. For  asymptotic theory analysis purposes related conditions are imposed to the noise sequences $\varepsilon_t$ such that the corresponding partial sum processes lie in the domain of attraction of functionals of Brownian motions. Thus asymptotic approximations are obtained based on the Ornstein-Uhlenbeck process.       
\end{remark}

Extending and verifying the existence of Cauchy limiting distribution theory in mildly explosive autoregression under various conditions on the innovation structure (stationary conditional heteroscedastic errors, anti-persistence errors - \cite{lui2021mildly}, non-Gaussian errors, serially correlated errors - \cite{lui2018mild}, errors with possibly infinite variance - \cite{liu2021quantile} and \cite{wang2022asymptotics}, mixing innovations - \cite{oh2018mildly} and \cite{liu2022mildly}) doesn't necessarily imply the presence of a unified theory across the spectrum of nonstationarity as per Assumption 1. In other words, the fact that there is a limiting distribution discontinuity across these different regions of the parameter space of $\theta_n$, despite that the separate Gaussian and Cauchy limit results are found to be robust against different properties in the innovation structure, still requires considering whether a unified distribution-free inference technique can be applied regardless of: \textit{(i)} the presence of a model intercept, \textit{(ii)} the innovation sequence properties and \textit{(iii)} the region of the parameter space. 

Furthermore, based on extensive empirical applications using the dataset of \cite{hardle2016tenet}, where cross-sectional specific  autoregressive and predictive regressions were fitted on stock returns with predictors the macroeconomic variable of the dataset, we have observed that in cases where the estimated autoregressive coefficient is above the unit boundary (i.e., on the mildly explosive side), the estimation method for the model coefficient of the predictive regression model remains unchanged if one uses the existing IVX instrumentation approach proposed by \cite{Magdalinos2009limit}  (see, also \cite{kostakis2015Robust}). This observation motivated us to consider an alternative instrumentation procedure when the estimated autoregressive coefficient is within that region is explosive.

\newpage

More precisely, the novel instrumentation procedure proposed by \cite{Magdalinos2022uniform} does exactly that, while theoretical and empirical illustrations are given for the linear autoregressive and predictive regression models. Specifically, the asymptotic theory of estimators, test statistics and corresponding confidence intervals can be constructed using a common distribution without imposing additional assumptions. These properties are found to hold specifically for autoregressive and predictive regression models based on a conditional mean functional form while the purpose of this paper is to extend those to the case of a conditional quantile functional form. 

The class \textbf{P.1} implies that $\theta_n$ is near to the unit boundary under there is no local-to-unity parametrization, thus the limit tends to $- \infty$. For the class \textbf{P.2} we have that $\theta_n \to 1$ when specified as $\theta_n = \left( 1 + \frac{c}{k_n}    \right)$, where $k_n = n^{\gamma}$ with $\gamma = 0$, $\gamma = 1$ or $\gamma \in (0,1)$. Thus, when $\theta_n = \left( 1 + \frac{c}{k_n}    \right)$, 
\begin{align}
\zeta_n := \underset{ n \to + \infty  }{ \mathsf{lim} } n \left( \theta_n - 1 \right) =  \underset{ n \to + \infty  }{ \mathsf{lim} } n \left[ \left( 1 + \frac{c}{k_n} \right) - 1 \right]   = \underset{ n \to + \infty  }{ \mathsf{lim} } \left\{ \frac{n}{k_n} \right\} c
\end{align}

Assume that $\textcolor{blue}{c < 0}$ then it holds that 
\begin{enumerate}
    
    \item If $k_n = n^{\gamma}$ with $\gamma = 1$ $( k_n = n )$, then $\textcolor{blue}{\zeta \equiv c \in \mathbb{R} }$, which falls in the class of \textbf{nearly-unstable} processes (near-nonstationary). 
    
    \item If $k_n = n^{\gamma}$ with $\gamma \in (0,1)$, then $\textcolor{blue}{\zeta = - \infty}$. In this case, although we are in the region of the "black circle" within the "blue zone" then we have a \textcolor{blue}{mildly integrated} process which falls in the class of \textbf{nearly-stable} processes (near-stationary). 
    
    \item If $k_n = n^{\gamma}$ with $\gamma = 0$ $( k_n = 1 )$, then $\textcolor{blue}{\zeta = - \infty}$ which falls in the class of \textbf{nearly-stable} processes (near-stationary).

\end{enumerate}

Assume that $\textcolor{red}{c > 0}$ then it holds that 
\begin{enumerate}
    
    \item If $k_n = n^{\gamma}$ with $\gamma = 1$ $( k_n = n )$, then $\textcolor{red}{\zeta \equiv c \in \mathbb{R} }$, which falls in the class of \textbf{nearly-unstable} processes (near-nonstationary). 
    
    \item If $k_n = n^{\gamma}$ with $\gamma \in (0,1)$, then $\textcolor{red}{\zeta = + \infty}$. In this case, although we are in the region of the "black circle" within the "red zone" then we have a \textcolor{red}{mildly explosive} process which falls in the class of \textbf{nearly-explosive} processes. 
    
    \item If $k_n = n^{\gamma}$ with $\gamma = 0$ $( k_n = 1 )$, then $\textcolor{red}{\zeta = + \infty}$ which falls in the class of \textbf{nearly-explosive} processes. 
    
\end{enumerate}

Assume that $c = 0$ then we have a pure unit root process.

\newpage

\underline{ \textbf{Persistence Classes:} }

\begin{itemize}
    
    \item[\textbf{I.}] \textbf{Pure Stationary Processes:} $\rho_n \to \rho \in (-1,1)$.   
    
    \
    
    \item[\textbf{II.}] \textbf{Near Unit Root Processes:} $\rho_n \to \rho = 1$ where $\rho_n$ is specified  as $\rho_n = \left( 1 + \frac{c}{k_n} \right)$, $k_n = n^{\gamma}$.
    
    For all the case below we can employe the original IVX instrument. 
    
    \begin{itemize}
        
        \item[\textbf{II.a.}]  \textbf{\textcolor{blue}{Near-stationary processes}}: $c <0$, $\gamma = 1$ (near-integrated)
        \begin{align}
            \underset{ n \to + \infty }{ \text{lim} } \rho_n := \underset{ n \to + \infty }{ \text{lim} } \left( 1 + \frac{c}{n}  \right) = 1 +  \underset{ n \to + \infty }{ \text{lim} } \left\{ \frac{c}{n} \right\} = \textcolor{blue}{ - \infty }, \ \text{when} \ c < 0.
        \end{align}

        \item[\textbf{II.b.}] \textbf{\textcolor{blue}{Mildly-integrated processes}}: $c <0$, $\gamma \in (0,1)$
        \begin{align}
            \underset{ n \to + \infty }{ \text{lim} } \rho_n := \underset{ n \to + \infty }{ \text{lim} } \left( 1 + \frac{c}{n}  \right) = 1 +  \underset{ n \to + \infty }{ \text{lim} } \left\{ \frac{c}{n} \right\} = \textcolor{blue}{ - \infty }, \ \text{when} \ c < 0.
        \end{align}
        
         \item[\textbf{II.c.}]  \textbf{\textcolor{blue}{Stationary processes}}: $c < 0$, $\gamma = 0$ 
        \begin{align}
            \underset{ n \to + \infty }{ \text{lim} } \rho_n := \underset{ n \to + \infty }{ \text{lim} } \left( 1 + c \right) = ( 1 + c ) \in \mathbb{R} 
        \end{align}
        
        \item[\textbf{II.$a^{\prime}$.}]  \textbf{\textcolor{red}{Near-explosive processes}}: $c > 0$, $\gamma = 1$ 
        \begin{align}
            \underset{ n \to + \infty }{ \text{lim} } \rho_n := \underset{ n \to + \infty }{ \text{lim} } \left( 1 + \frac{c}{n}  \right) = 1 +  \underset{ n \to + \infty }{ \text{lim} } \left\{ \frac{c}{n} \right\} = \textcolor{red}{ + \infty}, \ \text{when} \ c > 0.
        \end{align}
        
        \item[\textbf{II.$b^{\prime}$.}]  \textbf{\textcolor{red}{Mildly-explosive process}}: $c > 0$, $\gamma \in (0,1)$
        \begin{align}
            \underset{ n \to + \infty }{ \text{lim} } \rho_n := \underset{ n \to + \infty }{ \text{lim} } \left( 1 + \frac{c}{n}  \right) = 1 +  \underset{ n \to + \infty }{ \text{lim} } \left\{ \frac{c}{n} \right\} = \textcolor{red}{+ \infty}, \ \text{when} \ c > 0.
        \end{align}
        
        \item[\textbf{II.$c^{\prime}$.}]  \textbf{\textcolor{red}{Stationary process}}: $c > 0$, $\gamma = 0$
         \begin{align}
            \underset{ n \to + \infty }{ \text{lim} } \rho_n := \underset{ n \to + \infty }{ \text{lim} } \left( 1 + c \right) = ( 1 + c ) \in \mathbb{R} 
        \end{align}
        
        \item[\textbf{II.d.}] \textbf{Pure Unit Root processes}: $c = 0$ regardless if $\gamma = 1$ or $\gamma \in (0,1)$
        \begin{align}
            \underset{ n \to + \infty }{ \text{lim} } \rho_n := \underset{ n \to + \infty }{ \text{lim} } \left( 1 + \frac{0}{n}  \right) = 1 \in \mathbb{R}.
        \end{align}

    \end{itemize}
    
    \

     \item[\textbf{III.}] \textbf{Pure Explosive Processes:} $\rho_n \to \rho > 1$. 
     
     Thus, in this case to construct the IVX instrument we need to convert the process to be mildly explosive, i.e., so that is within the red region. 
        
\end{itemize}

\newpage 

\begin{example}
Consider the classification of nonstationarity given by \cite{benke2021nearly}.    
\begin{align}
X_k = \beta X_{k-1} + \varepsilon_k, \ \ \ X_0 = 0.    
\end{align}
\begin{itemize}
    
\item In the case that $|\beta | < 1$, the sequence $\left\{ \widehat{\beta}_n  \right\}_{ n \in \mathbb{N} }$ is asymptotically normal such that 
\begin{align}
\sqrt{n} \left( \widehat{\beta}_n  - \beta \right) \overset{d}{\to} \mathcal{N} \big( 0, 1 - \beta^2 \big), \ \ \ \text{as} \ \ \ n \to \infty. 
\end{align}

\item In the case that $| \beta | = 1$, then the sequence $\left\{ \widehat{\beta}_n  \right\}_{ n \in \mathbb{N} }$ has a limit distribution that depends on functionals of Brownian motion such that 
\begin{align}
n \left( \widehat{\beta}_n  - \beta \right) \overset{d}{\to} \frac{ \displaystyle \int_0^1 W(r) dW(r) }{  \displaystyle \int_0^1 W(r)^2 dr  },\ \ \ \text{as} \ \ \ n \to \infty.   
\end{align}

\item A sequence $\left( X^{ \beta^{(n)} } \right)_{ k \in \mathbb{N} }$ for $n \in \mathbb{N}$ of an AR(1) process is called nearly unstable if $\beta^{(n)} \to \beta$ as $n \to \infty$, where $| \beta | = 1$. Moreover, in the case that $\beta = 1$ and $\beta^{(n)} = \left( 1 + \frac{c}{n} \right)$ with some $c \in \mathbb{R}$, 
\begin{align}
n \left(  \widehat{\beta}_n^{ (n) } - \beta^{(n)} \right) \overset{d}{\to} \frac{ \displaystyle  \int_0^1 J(r) dW(r)  }{  \displaystyle  \int_0^1 J(r)^2 dr }, \ \ \ \text{as} \ n \to \infty.    
\end{align}
\end{itemize}
such that $\left\{ J_{ c } (r) \right\}_{ r \in [0,1] }$ is a continuous AR(1) process, that is, an OU process defined as a unique strong solution of the SDE such that 
\begin{align}
\begin{cases}
dJ_{c} (r) &= c J_{c} (r) dr + dW(r), \ \ \ r \in [0,1]
\\
Y_{(c)} (r) &= 0
\end{cases}  
\end{align}
Furthermore, if we consider $c$ to be a parameter instead of $\beta_{(n)}$ then the LSE of $c = n \left(  \beta^{(n)} - 1 \right)$ is given by the following expression:
\begin{align}
\hat{c}_n &= n \left( \widehat{\beta}_n^{(n)} -  1 \right) = n \left( \beta_n^{(n)} - \beta^{(n)} \right) + c,   
\\
\widehat{c}_n &\overset{d}{\to} \frac{ \displaystyle  \int_0^1 J_c(r) dW(r)  }{  \displaystyle  \int_0^1 J_c(r)^2 dr }  + c  
=
\frac{ \displaystyle  \int_0^1 J_c(r) dJ_c (r)  }{  \displaystyle  \int_0^1 J_c(r)^2 dr },   
\end{align}
where the limit distribution above turns out to be the maximum likelihood estimator (MLE) of the parameter $c$ for the expression that gives the stochastic differential equation based on a sample $\left\{ J_{ c } (r) \right\}_{ r \in [0,1] }$. 
\end{example}

\newpage

\begin{example}
Consider the following processes:
\begin{align}
x_t &= \mu + X_t
\\
X_t &= \rho X_{t-1} + \varepsilon_t 
\end{align}
which implies that 
\begin{align}
x_t = \mu \big( 1 - \rho_n \big) + \rho_n x_{t-1} + \varepsilon_t     
\end{align}
Moreover, denote with 
\begin{align*}
x_t^{\mu} = \big(  x_t - \bar{x}_t \big), \ \ \bar{x}_t = \frac{1}{n} \sum_{t=1}^n x_t, \ \ x_{t-1}^{\mu} = \big(  x_{t-1} - \bar{x}_{t-1} \big), \ \ \bar{x}_{t-1} = \frac{1}{n} \sum_{t=1}^n x_{t-1}.
\end{align*}
Equivalently, for $\mu \neq 0$, we consider the utoregressive model $x_t^{\mu} = \rho_n x_{t-1}^{\mu} + \varepsilon_t^{\mu}$, which implies
\begin{align}
\hat{\rho}_n = \left(  \sum_{t=1}^n x_{t-1}^{\mu} \right)^{-1} \left(  \sum_{t=1}^n x_{t-1}^{\mu}  x_{t}^{\mu}  \right) , \ \ \ \hat{\varepsilon}_t^{\mu} = x_{t}^{\mu}  - \hat{\rho}_n x_{t-1}^{\mu}  
\end{align}
\end{example}
Therefore in this article we focus on the quantile autoregressive model with a model intercept the particular modeling framework motivate us to investigate the interplay between near-zero intercept\footnote{Notice that  since the initial condition of the autoregressive process corresponds to the boundary condition of an ordinary differential equation problem, then the stochastic solution, that is, the limiting distribution under equilibrium conditions will depend on the initial condition (see, \cite{saxena1982estimation}).} and stochastic persistence by deriving the asymptotic distribution. 

\medskip

\begin{assumption}
\label{Assumption1}
The distribution $F_{\varepsilon}$ is in the domain of attraction of a stable law indexed with $\upalpha \in (0,2)$, which has a strictly positive density. When $\upalpha = 2$ then it corresponds to the case of innovation sequences from a distribution function with finite variance. Furthermore, 
\begin{align}
\mathbb{E} \left( \varepsilon_1 \right) = 0,  \ \ \mathsf{Var} \left( \varepsilon_1 \right) = \sigma^2_{\varepsilon }, \ \ \ \text{and} \ \ \  \mathbb{E} \left| \varepsilon_1 \right|^{2+m} < \infty, \ \text{for some} \ m > 0.   
\end{align} 
\end{assumption}
When $c = 0$ then the AR(1) model is a random walk model with stable innovations. Therefore, under the assumption that $c = 0$, which implies the presence of an integrated process the asymptotic behaviour of the ordinary least squares estimator can be obtained based on Brownian motion functionals. Define with $U_{\upalpha}(r)$ and $V_{\upalpha}(r)$ to be L\'evy processes on the space of functions $\mathcal{D}[0,1]$. 

\medskip

\begin{lemma}
\label{Lemma1}
Suppose that $\varepsilon_t$ satisfies Assumption \ref{Assumption1}. Then as $n \to \infty$ it holds that 
\begin{align}
\left( \sum_{t=1}^{ \floor{nr} } \frac{\varepsilon_t}{a_n}, \sum_{t=1}^{ \floor{nr} } \frac{\varepsilon_t^2}{a_n^2} \right) \Rightarrow \bigg( U_{\upalpha} (r), V_{\upalpha}(r) \bigg)
\end{align}
where $\big( U_{\upalpha} (r), V_{\upalpha}(r) \big)$ is a Levy process in $\mathcal{D}[0,1]^2$ with index $\upalpha \in (0,2)$. 
\end{lemma}

\newpage

\begin{remark}
A general class of one-dimensional stochastic processes, are the so-called L\'evy processes. Similar to Wiener processes, L\'evy processes have right continuous paths with left limits, are initiated from the origin and both have stationary and independent increments (\cite{kyprianou2014fluctuations}). Under the \textit{i.i.d} innovation assumption it can be shown that $V_{\upalpha} (r) = \int_0^r \left( dU_{\upalpha} (s) \right)^2 = \big[ U_{\alpha}, U_{\alpha} \big]_{ r \in [0,1] }$ which is the quadratic variation of the Levy process $U_{\upalpha} (r)$. Furthermore, $V_{\alpha} (r)$ is a stochastic integral. When $\upalpha = 2$, which corresponds to the finite variance case, it holds that $ U_{\upalpha} (r) \equiv W(r)$ for some $0 \leq r \leq 1$, the standard Brownian motion, and  $V_{\upalpha} (r) = \big[ W, W  \big]_r = r$ (see, also \cite{cramer1951contribution}).  
\end{remark}
Any asymptotic results followed by Lemma \ref{Lemma1} coincide with the standard random walk asymptotics for finite variance models (we focus in the case $\upalpha = 2$). Then as $n \to \infty$, it holds that
\begin{align}
n \left( \hat{\rho}_{n,c} - 1 \right) \Rightarrow \frac{ \displaystyle \int_0^1 W(r) dW(r) }{ \displaystyle \int_0^1 W^2(r) dr }.   
\end{align}  
Since $\rho_n = \left( 1 + \frac{c}{n^{\gamma}}\right)$ under the stationary condition which implies that $0 < \rho_n < 1$ then it holds that $- n^{\gamma} < c < 0$. The OLS estimate $\hat{\rho}$ is $n-$consistent, that is, $n \left( \hat{\rho} - \rho \right)$ has a nondegenerate limit distribution depending on $c$, while $\hat{\mu}$ is $\sqrt{n}-$asymptotically normal. As a result, $\hat{c} = - n( 1 - \hat{\rho} )$ is the  OLS estimate of the coefficient of persistence $c$, which is not consistent. \cite{mikusheva2012one} shows   
\begin{align}
\left( \hat{c} - c \right) = \frac{ \displaystyle \int_0^1 J_c(r) dW(r) }{ \displaystyle \int_0^1 J_c^2 (r) dr  }.
\end{align}
The particular asymptotic result demonstrates the well-known conjecture that the OLS estimate of the nuisance parameter of persistence, $c$, is not consistent. In fact, $\hat{c}$ is asymptotically highly biased to the left, thus the estimated model looks more stationary that it actually is. Roughly speaking, the assumption regarding the dependence structure for the disturbance term can affect the limit theory of estimators. In particular, serial correlation in the errors induces an asymptotic bias for $\hat{\rho}_n$ and contributes to the bias of the Gaussian limiting distribution. For the asymptotic theory we assume that $\mathcal{D}[0,1]$ is endowed with the Skorokhod topology such that any partial sum processes are measurable for the associated Borel $\sigma-$algebra under the absence of serial correlation. Imposing assumptions regarding the properties of the disturbance term $\varepsilon_t$ can imply different asymptotic behaviour for estimators given certain modelling conditions. According to \cite{werker2022semiparametric} the usual procedures established in the literature thus far are based on the assumption of Gaussian innovations and, while their validity has been established under weak assumptions, the asymptotic power of all these procedures cannot go beyond the Gaussian power envelope. Relaxing the Gaussianity assumption requires to apply semiparametric estimation methodologies is beyond the scope of our study. Thus, the \textit{i.i.d} innovation sequence assumption with finite variance,  simplifies the representation of the necessary regularity conditions. 

\newpage


\subsection{Quantile Conditional Estimation}

In this Section we discuss in more details the estimation procedure for the quantile autoregressive time series that accommodates the parametrization of the autoregression  coefficient with respect to moderate deviations from the unit boundary. We first introduce the quantile estimation method\footnote{A complete treatment of limit results for quantile regressions can be found in the book of \cite{Koenker2005}.} to obtain parameter estimates and then establish the asymptotic theory of this estimator for the autoregressive model denoted with $y_t = \mu + \rho y_{t-1} + u_t$ (with a slight change of notation).  

Denote with $\mu( \uptau )$ and $\rho_n ( \uptau )$ to be the $\uptau-$quantile dependent parameters, which are determined based on a conditional quantile functional form as below:
\begin{align}
\mathsf{Q}_{y_t} \left( \uptau | \mathcal{F}_{t-1} \right) &:=  F^{-1}_{ y_t | \boldsymbol{x}_{t-1} } (\uptau) \equiv \mu (\uptau) + \rho_{c} (\uptau) y_{t-1}. 
\\
F_{ y_t | \boldsymbol{x}_{t-1} } (\uptau) &:= \mathbb{P} \bigg( y_t \leq \mathsf{Q}_{y_t} \left( \uptau | \mathcal{F}_{t-1} \right) \big| \mathcal{F}_{t-1} \bigg) \equiv \uptau.
\end{align}
for some $\uptau \in (0,1)$, where $\uptau$ denotes the quantile level within a compact set $(0,1)$. 

Denote the parameter vector with $\boldsymbol{\vartheta} ( \uptau ) = \big( \mu  ( \uptau ), \rho_c ( \uptau ) \big)^{\top}$ and $\boldsymbol{X}_t = \boldsymbol{D}_n^{-1} \big( 1, y_{t-1} \big)$, where $\boldsymbol{D}_n$ is the normalization matrix which includes the different convergence rates for the model intercept vis-a-vis the slope coefficient. Then, from \cite{koenker1978regression} and \cite{koenker1987estimation} the quantile regression estimator is obtained via the following  optimization function 
\begin{align}
\label{the.problem}
\widehat{ \boldsymbol{\vartheta} }_n \big( \uptau \big) 
:= 
\underset{ \boldsymbol{\vartheta} ( \uptau )  }{ \mathsf{arg \ min} } \ \sum_{t=1}^n \varrho_{\uptau} \bigg( y_{t} - \boldsymbol{\vartheta} ( \uptau )^{\top} \boldsymbol{X}_{t} \bigg).
\end{align}
such that $\widehat{ \boldsymbol{\vartheta} }_n \big( \uptau \big) \equiv \boldsymbol{D}_n \left( \hat{\mu}_n ( \uptau ) - \mu ( \uptau ), \hat{\rho}_{n,c} ( \uptau ) - \rho_{c}( \uptau ) \right)$.  
Moreover, we denote with $\psi(\mathsf{u})$ to be the left derivative of $\varrho(\mathsf{u})$. In particular, when $\psi( \mathsf{u} ) := \mathsf{u}$ is the identity function then the $\widehat{ \boldsymbol{\vartheta} }_n$ corresponds to the least squares estimator, while when $\psi(\mathsf{u}) := ( 1 - \uptau )$ for $\mathsf{u} \leq 0$ and  $\psi(\mathsf{u}) = \uptau$ for $\mathsf{u} > 0$, then it corresponds to the optimization function and $\widehat{ \boldsymbol{\vartheta} }_n$ is the quantile-dependent estimator.  
\begin{assumption}
\label{Assumption2}
Suppose that $\mathbb{E} \big[ \psi ( u_t ( \uptau )  ) \big] = 0$ and consider the random variable which corresponds to the first derivative around some parameter $\theta \in \mathbb{R}$ such that 
\begin{align}
\xi := \left| \frac{ \partial }{ \partial \theta } \mathbb{E} \bigg[ \psi \bigg( u_1 ( \uptau ) - \theta \bigg) \bigg] \right|_{ \theta = 0 }, \ \ \text{where} \   \xi  \neq 0.
\end{align}
where $\mathbb{E}| \psi \left(  u_1 ( \uptau ) \right) |^{2 + m} < \infty$, for some $m > 0$.
\end{assumption}
Assumption \ref{Assumption2} ensures that the first derivative is Lipschitz continuous and bounded which corresponds to the first derivative for the expectation of the check function as a random variable evaluated within the neighbourhood of the true parameter vector $\theta = 0$. Furthermore, due to the fact that the quantile autoregressive model we consider in this paper corresponds to a possibly nonstationary time series model, then the asymptotic theory of estimators and corresponding test statistics depends on Brownian motion functionals as introduced with Assumption \ref{Assumption3} below.

\newpage 

\begin{assumption}
\label{Assumption3} 
The following conditions for the innovation sequence hold:
\begin{itemize}
\item[\textbf{(\textit{i})}] The sequence of stationary conditional probability distribution functions denoted with $\big\{ f_{ \varepsilon_t (\uptau), t-1}(.) \big\}$ evaluated at zero with a non-degenerate mean function $f_{ \varepsilon_t (\uptau)  }(0) := \mathbb{E} \left[  f_{ \varepsilon_t (\uptau), t-1}(0) \right] > 0$, that satisfies a Functional Central Limit Theorem (FCLT) expressed as below
\begin{align}
\frac{1}{ \sqrt{n} } \sum_{t=1}^{ \floor{nr} } \left( f_{  \varepsilon_t (\uptau), t-1}(0) - \mathbb{E} \left[  f_{ \varepsilon_t (\uptau), t-1}(0) \right] \right) \Rightarrow B_{ f_{  \varepsilon_t (\uptau) } } (r), \ \text{with} \ r \in (0,1).
\end{align}
\item[ \textbf{(\textit{ii})} ]  For each $t$ and $\uptau \in (0,1)$, $f_{ \varepsilon_t (\uptau), t-1}(.)$ is uniformly bounded away from zero with a corresponding conditional distribution function $F_t(.)$ which is absolutely  continuous with respect to Lebesgue measure on $\mathbb{R}$ (see, \cite{neocleous2008monotonicity}, \cite{goh2009nonstandard} and \cite{lee2016predictive}).
\end{itemize}
\end{assumption}

\begin{remark}
Assumption \ref{Assumption3} gives necessary and sufficient conditions for a functional central limit theorem to hold for the corresponding innovation sequence based on the conditional quantile functional form, which is instrumental for deriving the asymptotic behaviour of the quantile-dependent estimators under nonstationarity based on Brownian motion functionals. 
\end{remark}
Therefore, to obtain the model estimates based on the optimization problem \eqref{the.problem} we apply the Taylor expansion to the check function, such that for a given parameter $\boldsymbol{\delta} ( \uptau )$ it holds that 
\begin{align}
\label{B.approx}
\varrho_{\uptau} \bigg( \varepsilon_{t} - \boldsymbol{\delta} ( \uptau )^{\top} \boldsymbol{X}_{t} \bigg) 
= 
\varrho_{\uptau} ( \varepsilon_{t} ) - \boldsymbol{\delta} ( \uptau )^{\top} \psi ( \varepsilon_{t} ) + \varphi_t \big( \boldsymbol{\delta} ( \uptau ) \big).
\end{align}

\begin{remark}
Notice that for instance the t-ratio for $\rho_{n,c}$ is defined by $\sqrt{ \sum_{t=1}^n y_{t-1}^2 } \left( \hat{\rho}_{n,c} - \rho_c \right)$, thus to obtain the limiting distribution of the $t-$test we need to obtain an asymptotic expression for the normalized centered estimator $\left( \hat{\rho}_{n,c} - \rho_c \right)$. Furthermore, for sequences such that $\underset{ n \to \infty }{ \mathsf{lim} } n \left( 1 - \rho_n \right) = 0$, the nearly unstable model behaves asymptotically like the strictly unstable model in which case $\rho = 1$.
\end{remark}

\subsection{Large-Sample Theory}
\label{Section3}

In this section we present the main asymptotic theory results while detailed proofs can be found in the Appendix of the paper. Some important aspects worth emphasizing again is that while the case of near-integrated (NI) processes, such that $c < 0 $ and $\gamma = 1$, has been considered before in quantile autoregressive time series (see, \cite{chan2006quantile}) as well as the case of mildly integrated (MI) such that $c < 0 $ and $\gamma \in (0,1)$, the mildly explosive (ME) such that $c > 0 $ and $\gamma \in (0,1)$ and the explosive case, such that $c > 0 $ and $\gamma =1$ has not been widely explored before. In particular, all aforementioned cases which correspond to different regions of the parameter space, overcome the singularity problem; the region that the underline stochastic process does not have a solution. Using the local-to-unity parametrization in autoregressive processes overcomes this gap by considering the limiting distribution\footnote{The form of the noncentrality parameter of the $\chi^2-$distribution depends on the initial condition and the form of the autoregressive parameter of the model (see, Theorem 2.2. of \cite{Jian2022}). } for the whole parameter space regardless of the existence of a limit singularity.

\newpage

Using the local-to-unity asymptotics our aim is to derive a nuisance-parameter-free limiting distribution which can facilitate statistical inference. Therefore, by decomposing the underline stochastic processes into  components which include a  predictable quadratic variation, allows us to obtain a self-normalized martingale sequence, which is especially useful when deriving the limiting distribution of Wald-type statistics. In other words, Wald statistics constructed with a variance estimator which induced by the predictable quadratic variation ensures that the self-normalization property holds. When asymptotic theory of the autoregressive parameter that corresponds to a particular region of the parameter space is nonstandard, then this implies that the limit distribution of the MLE of the model can be represented as the MLE of a parameter of a process satisfying the stochastic differential equation for the OU process (see, \cite{benke2021nearly}). Thus, the limiting distributions of the model estimator and the corresponding test statistic are functions of the local-to-unity parameter which is not consistenty estimable since $\left( \hat{c} - c \right) = \mathcal{O}_p (1)$. 

Although in this article we only consider point inference procedures, in the case of interval inference, that is, when is concerned with the construction of confidence interval for unknown model parameters then several studies discuss the possibility of obtaining uniform inference procedures when the nonstationary autoregressive model is expressed using the local-to-unity parametrization (see, \cite{mikusheva2007uniform}, \cite{phillips2014confidence} and \cite{Magdalinos2022uniform}). In this article we focus on bridging the gap in the robust and uniform inference literature on quantile autoregressions and quantile predictive regression models by considering that the instrumentation procedure depends on the region of the parameter space. In other words, although the original IVX method proposed by \cite{PM2009econometric} and \cite{kostakis2015Robust} it is found to be robust in LUR moderate deviations from the unit boundary, it is clear than in explosive and mildly explosive regimes the same IVX instrumentation might not work so well. 
\begin{theorem}[\cite{chan2006quantile}]
Assume that Assumption \ref{Assumption0}-\ref{Assumption3}  hold and that the autoregression  coefficient is expressed as $\rho_{n,c} = \left( 1 + \frac{c}{k_n}  \right)$. Then, the following limit result hold  
\begin{align}
\boldsymbol{D}_n \left(  \hat{ \boldsymbol{\vartheta} }_n ( \uptau ) - \boldsymbol{\vartheta} ( \uptau )  \right) &\overset{ d }{ \to }  \frac{1}{ f_{\varepsilon} \big( F_{\varepsilon}^{-1} ( \varepsilon_t ( \uptau ) ) \big) } \Sigma^{-1} \left(  W ( \uptau ; 1 ) , \int_0^1 J(s) d W(  \uptau, s ) \right)^{\prime},
\\
n \bigg( \hat{\rho}_{n,c}( \uptau ) - \rho_{n,c} ( \uptau ) \bigg)
&\overset{ d }{ \to }  
\frac{1}{ f_{\varepsilon} \big( F_{\varepsilon}^{-1} ( \varepsilon_t ( \uptau ) ) \big) } \frac{ \displaystyle \int_0^1 J_1(s) dW( \uptau, s ) - W(\uptau, 1 ) \int_0^1 J_1(s) ds  }{ \displaystyle \int_0^1 J^2_1(s) ds - \left( \int_0^1 J_1(s) ds \right)^2 }
\end{align}
where $\boldsymbol{D}_n = \mathsf{diag} \left( \sqrt{n}, n \right)$ and $\boldsymbol{\vartheta}  ( \uptau ) = \big( \mu( \uptau ), \rho_{n,c} ( \uptau ) \big)$ such that 
\begin{align}
\boldsymbol{\Sigma} 
:= 
\int_0^1 \big( 1, J_1(s) \big)^{\prime} \big( 1, J_1(s) \big) ds 
\equiv
\begin{bmatrix}
1 &  \displaystyle \int_0^1  J_1(s) ds  
\\
\displaystyle \int_0^1  J_1(s)^{\prime} ds & \displaystyle \int_0^1  J_1(s) J_1(s)^{\prime} ds 
\end{bmatrix}.
\end{align}
\end{theorem}

\newpage

\begin{align}
\left\{ \sum_{t=1}^n y_{t-1}^2 - \left( \sum_{t=1}^n y_{t-1}  \right)^2 \right\}^{1/2} \bigg( \hat{\rho}_{n,c}( \uptau ) - \rho_{n,c} ( \uptau ) \bigg)
&\overset{ d }{ \to } 
\mathcal{N} \left( 0, \frac{ \uptau ( 1- \uptau) }{ f_{\varepsilon}^2 \big( F_{\varepsilon}^{-1} ( \varepsilon_t ( \uptau ) ) \big)   }  \right).
\end{align}

\subsubsection{Limit theory for near-stationary case}

In a similar spirit as in the framework proposed by \cite{Phillips2007limit}, in order to establish the limit theory of the autoregression  coefficient within our modelling environment,  we consider the asymptotic behaviour of the sample moments that appear in the quantile-dependent estimator separately. However, in contrast to the ordinary least squares estimation, when the model parameters are estimated using the conditional quantile functional form, we employ standard approximation methods (such as the Bahadur representation) from the quantile regression literature to obtain analytical expressions for the quantities of interest. Specifically, in the near-stationary case, which implies that $c < 0$, the limit to the unit boundary is approached from the left of the triangular array. Furthermore, due to the different convergences rate of the model intercept versus the slope parameter we employ the normalization matrices $\boldsymbol{D}_n$ and $\boldsymbol{B}$ as defined below
\begin{align}
\boldsymbol{D}_n = 
\begin{pmatrix}
\sqrt{n} & 0 
\\
0 & \sqrt{n k_n}
\end{pmatrix}, \ \ \
\boldsymbol{B} = 
\begin{pmatrix}
1 & 0 
\\
0 &  \sigma^2 / (-2c)  
\end{pmatrix}
\end{align}
where $k_n = n^{\gamma}$ and $\gamma \in (0,1)$. The $n^{-1 / 2 }$ convergence rate corresponds to the model intercept while when $k_n = n^{\gamma}$, then the autoregression  parameter of the model has a convergence rate of $n^{ - \frac{1 + \gamma}{2} }$ which is also the rate of convergence that corresponds to a mildly integrated process. Furthermore, in empirical applications in practise we do not know a prior whether the expression $\sqrt{n} \left( \hat{\rho}_n - \rho  \right)$ is positive or negative. However, since we do not partition the parameter space accordingly, the asymptotic theory mainly focus on the near-integrated case and does not cover the mildly explosive or pure explosive since $c < 0$. 

\begin{theorem}
\label{theorem1}
Under Assumptions \ref{Assumption1}-\ref{Assumption3}, 
\begin{align}
\big( \hat{\mu}_n, \hat{\rho}_{n,c} \big)^{\top} = \big( \mu, \rho_{n,c} \big)^{\top} + \frac{ ( \boldsymbol{B} \boldsymbol{D}_n )^{-1} }{ \xi } \sum_{t=1}^n \psi \big( \varepsilon_t \big) \boldsymbol{X}_t^{\top} + o_p(1).
\end{align}
In particular, when $\psi( \mathsf{u} ) = \big( \uptau - \mathbf{1} \left\{ \mathsf{u} \leq 0 \right\} \big)$ corresponds to the quantile regression and therefore the above expression reduces to 
\begin{align}
\begin{pmatrix}
\hat{\mu}_n ( \uptau )
\\
\hat{\rho}_{c,n} ( \uptau )
\end{pmatrix}
= 
\begin{pmatrix}
\mu_n ( \uptau )
\\
\rho_{c,n} ( \uptau )
\end{pmatrix}
+
\frac{ \left( \boldsymbol{B} \boldsymbol{D}_n \right)^{-1} }{ f_{\varepsilon} \big( F_{\varepsilon}^{-1} ( \uptau ) \big) } \sum_{t=1}^n \bigg( \uptau - \mathbf{1} \left\{ \varepsilon_t \leq F_{\varepsilon}^{-1} ( \uptau )   \right\} \bigg)
\begin{pmatrix}
\displaystyle \frac{1}{\sqrt{n}}
\\
\\
\displaystyle \frac{ y_{t-1} }{\sqrt{n k_n}}
\end{pmatrix}
+ o_P(1).
\end{align}
where $f_{\varepsilon}(x)$ and $F_{\varepsilon}(x)$ denote the probability and cumulative density functions of $\varepsilon_1$, respectively. 
\end{theorem}

\newpage 

\begin{remark}
Theorem \ref{theorem1} provides a Bahadur representation for the parameter vector of the quantile autoregressive time series which includes a model intercept and a slope. In particular for M-regressions a necessary requirement for the functional form is to include a model intercept which can be different than zero. Moreover, the given limit results are employed to derive the asymptotic behaviour of model parameters based on moderate deviations from the unit boundary on the stationary region as summarized by the next theorem. Then, the robust estimation of the sparsity coefficient which depends on the  kernel density function can be improve the accuracy of the quantile-dependent model estimates.  
\end{remark}

\begin{theorem}
\label{theorem2}
Under Assumptions \ref{Assumption1}-\ref{Assumption3},
\begin{align}
\boldsymbol{D}_n \left( \big( \hat{\mu}_n, \hat{\rho}_{n,c}   \big) -  \big( \mu_n, \rho_{n,c} \big) \right) \overset{ d }{ \to } \mathcal{N} \left( 0, \frac{ \boldsymbol{B}^{-1} \times \mathbb{E} \big[ \psi^2 ( \varepsilon_1 ) \big] }{ \xi  } \right).
\end{align}
In particular, it follows that 
\begin{itemize}

\item[\textit{(i)}] If $\psi( \mathsf{u} ) = \big( \uptau - \mathbf{1} \left\{ \mathsf{u} \leq 0 \right\} \big)$, and the pdf $f(\mathsf{u})$ of $\varepsilon_1$ exists and satisfies $f_{\varepsilon} \big( F_{\varepsilon}^{-1} (\uptau) \big) > 0$, then  
\begin{align}
\frac{ \hat{\rho}_{n,c} ( \uptau ) - \rho_{c} ( \uptau ) }{ \sqrt{n k_n} } \overset{ d }{ \to } \mathcal{N} \left( 0, \frac{-2c}{ \sigma^2 }  \frac{\uptau(1 -\uptau)}{ f_{\varepsilon}^2 \big( F_{\varepsilon}^{-1} (\uptau) \big) }  \right).
\end{align}

\item[\textit{(ii)}] If $\psi( \mathsf{u} ) = \mathsf{u}$, then  
\begin{align}
\frac{ \hat{\rho}_{n,c} ( \uptau ) - \rho_{c} ( \uptau ) }{ \sqrt{n k_n} } \overset{ d }{ \to } \mathcal{N} \left( 0, -2c \right).
\end{align}

\end{itemize}
\end{theorem}

\medskip

\begin{remark}
The limit results given by Theorem \ref{theorem1} summarize the joint asymptotic behaviour of the model intercept and slope from moderate deviations from the unit boundary on the stationary region. In order to prove the above asymptotic results, we employ standard arguments introduced by \cite{pollard1991asymptotics} for optimization of convex function relevant to quantile processes. Specifically, by the convexity lemma, if the finite-dimensional distributions of $\Omega_n (v)$ converge weakly to those of $\Omega (v)$, and $\Omega (v)$ has a unique minimum, then the convexity of $\Omega_n (v)$ implies that $\hat{v}$ converges in distribution to the minimizer of $\Omega (v)$. In other words, $\hat{\rho}_n( \uptau )$ is shown to be weakly consistent, thus to prove that the estimator is asymptotically normally distributed we restrict the spaces $\mathcal{B}$ to shrinking neighbourhoods around the true value of the parameter $\rho ( \uptau )$ in order to avoid possible local minima. To do this, we can define the restricted space $\mathcal{B}_{a} = \big\{ \rho_n \in \mathcal{B} \norm{ \beta - \beta ( \uptau ) } \leq a_n \big\}$ where $\left\{ a_n \right\}$ is some positive sequence.
\end{remark}

\newpage

\subsubsection{Limit theory for near-explosive case}

The near-explosive case corresponds to the nuisance parameters of persistence $c >0$ and the exponent rate $\gamma \in (0,1)$ or $\gamma = 1$. In particular for the linear autoregressive process $y_t = \theta y_{t-1} + \varepsilon_t$ and an explosive root such that $| \theta | > 1$, a Cauchy limit theory can be derived for the OLS estimator $\hat{\theta}$ as 
\begin{align}
\frac{ \theta^n  }{ \theta^2 - 1 } \left( \hat{\theta}_n - \theta \right) \Rightarrow \mathcal{C}, \ \ \ \text{as} \ n \to \infty.
\end{align} 

More precisely, the seminal study of \cite{anderson1959asymptotic} provides examples demonstrating that central limit theory does not apply and the asymptotic distribution of the least squares estimator depends by the distributional assumptions imposed on the innovations which makes inference procedures specifically for purely explosive autoregressions more challenging (see, \cite{magdalinos2012mildly}). Furthermore, in this direction, \cite{Phillips2007limit} consider autoregressive processes under the moderate deviations framework by employing the local-to-unity parametrization for the autoregression coefficient such that $\theta_n = \left( 1 + \frac{c}{n^{\gamma} } \right), \gamma \in (0,1)$. Therefore, in this case under the assumption of $\textit{i.i.d}$ innovations with finite second moments the following least squares regression theory was proved
\begin{align}
\frac{1}{2c} n^{\gamma} \theta_n^n \left( \hat{\theta}_n - \theta \right) \Rightarrow \mathcal{C}, \ \ \ \text{as} \ n \to \infty.
\end{align} 
On the other hand, for the pure explosive root case such that $| \theta | > 1$ then, the limit distribution of $\big( \hat{\theta} - \theta \big)$ is standard Cauchy if it is normalized with $\theta^n / ( 1 - \theta^2 )$. However, the limit distribution depends on the distribution of the noise, as was pointed out by \cite{anderson1959asymptotic}, and hence no central limit theorem applies on the explosive side.  Moreover, from empirical data financial applications it can be observed that the parameter $\theta$ tends to 1 with increasing sample size. 

To accommodate this observation, $\theta = \theta_n$ is allowed to depend on $n$, the number of observations, such that $\theta_n \to 1$ as $n \to \infty$. The process is then referred to as near-integrated. Depending on whether $\theta_n < 1$ or $\theta_n > 1$, it is called near-stationary or mildly explosive. Furthermore, \cite{Phillips2007limit} investigated the general parameter case in the near-integrated setting assuming that $\theta_n \to 1$ with a rate slower than $1 / n$, the so-called moderate deviations from unity. All aforementioned approaches operate under the assumption of a finite variance along with independent, identically distributed or weakly dependent errors. However, it can be proved that the serial coefficient $\hat{\theta}_n - \theta_n$ has, under a suitable normalization, a limit that consists of a fraction of two independent strictly stable random variables. 

Therefore, specifically for the quantile autoregressive time series model we consider in our study we employ the following normalization matrices. 
\begin{align}
\boldsymbol{D}_n = 
\begin{pmatrix}
\sqrt{n} & 0 
\\
0 &  \rho_{n, c}^n k_n  
\end{pmatrix},
\ \ \
\boldsymbol{B} = 
\begin{pmatrix}
1 & 0 
\\
0 &  \mathcal{Z}_3^2 / (2c)  
\end{pmatrix}
\end{align}
where $\mathcal{Z}_3$ is some normal random variable to be defined below. 

\medskip

\begin{theorem}
Under the Assumptions \ref{Assumption0}-\ref{Assumption3} it holds that, 
\begin{align}
\left( \hat{\mu}, \hat{\rho}_{n,c} \right)^{\top} = \left( \mu, \rho_{n,c} \right)^{\top} + \frac{ \left( \sum_{t=1}^n \boldsymbol{X}_t  \boldsymbol{X}_t^{\top} \boldsymbol{D}_n \right)^{-1} }{ \xi  } \sum_{t=1}^n \psi \big( \varepsilon_t  \big) \boldsymbol{X}_t^{\top} + o_p(1).
\end{align}
\end{theorem}

\bigskip

\begin{theorem}
\label{theorem5}
Under Assumptions \ref{Assumption0}-\ref{Assumption3} it holds that, 
\begin{align}
\boldsymbol{D}_n \big( \left( \hat{\mu}, \hat{\rho}_{n,c} \right) - \left( \mu, \rho_{n,c} \right) \big)^{\top} \overset{ d }{ \to } \frac{1}{ \xi  } \boldsymbol{B}_n^{-1} \big( \mathcal{Z}_1, \mathcal{Z}_2 \mathcal{Z}_3     \big)^{\top}.
\end{align}
In particular, it follows that 
\begin{itemize}

\item[\textit{(i)}] If $\psi( \mathsf{u} ) = \big( \uptau - \mathbf{1} \left\{ \mathsf{u} \leq 0 \right\} \big)$, and the pdf $f(\mathsf{u})$ of $\varepsilon_1$ exists and satisfies $f_{\varepsilon} \big( F_{\varepsilon}^{-1} (\uptau) \big) > 0$, then  
\begin{align}
\frac{ \hat{\rho}_{n,c} ( \uptau ) - \rho_{c} ( \uptau ) }{  \rho_n k_n  } \overset{ d }{ \to } \frac{2c}{ f_{\varepsilon} \left( F_{\varepsilon}^{-1} ( \uptau ) \right) } \frac{ \mathcal{Z}_2 }{ \mathcal{Z}_3 }.    
\end{align}

\item[\textit{(ii)}] If $\psi( \mathsf{u} ) = \mathsf{u}$, then  
\begin{align}
\frac{ \hat{\rho}_{n,c} ( \uptau ) - \rho_{c} ( \uptau ) }{  2 c \rho_n^n k_n  } \overset{ d }{ \to } \frac{ \mathcal{Z}^{*}_2 }{ \mathcal{Z}^{*}_3 }.  
\end{align}

\end{itemize}
\end{theorem}

\medskip

Therefore, Theorem \ref{theorem5} verifies that we indeed obtain the equivalent asymptotic theory results in comparison to the linear autoregressive time series model. Specifically, the autoregression  coefficient of the nonstationary quantile autoregressive time series model converge into a Cauchy random variate in the case of mildly explosive processes (see, also \cite{aue2007limit}, \cite{Phillips2007limit}, \cite{magdalinos2012mildly} and \cite{lee2018limit}).  

\medskip

\begin{remark}
As we see from Theorem \ref{theorem5}, part 2, the limiting distribution of the normalized and centered estimator is Cauchy, similar to Theorem 4.3 of \cite{Phillips2007limit}. As a matter of fact when we replace $\rho_n$ by $\rho_{n,c} = \left( 1 + \frac{c}{k_n} \right)$ we obtain that $\left( \rho^2 - 1 \right) = \frac{ 2c }{ k_n } \big[ 1 + o(1) \big]$. Hence, we see that the normalizations in the Theorem above and the expression derived by \cite{white1958limiting} are asymptotically equivalent as $n \to \infty$. Furthermore, the asymptotic theory for the case of moderate deviations from the unity boundary is not restricted to Gaussian processes. More specifically, the Cauchy limit result applied for $\rho_{n,c} = \left( 1 + \frac{c}{k_n} \right)$ and innovations $\varepsilon_t$ with finite second moment (e.g., innovations with stable law of attraction). On the other hand, the main difference between the mildly explosive processes given by Theorem \ref{theorem5} above and explosive autoregressions with $| \rho | > 1$, occurs due to the different convergence rates of these two cases. In particular, in the case of mildly explosive processes we define the convergence rate such that $k_n = n^{\gamma}$ for some $\gamma \in (0,1)$ while for the case of moderately explosive processes we define with $k_n = n^{\gamma}$ for some $\gamma > 1$. 
\end{remark}

\newpage

\subsection{Testing Linear Hypotheses}

Consider the autoregressive model 
\begin{align}
y_t = \rho y_{t-1} + \varepsilon_t, \ \ \ \ t = 1,..., n, 
\end{align}
such that $\rho \in [-1,1]$, is within the stationary region. Then, the usual testing hypothesis of interest is such that, $\mathbb{H}_0: \rho = \rho_0$. In particular, \cite{dickey1979distribution} showed that the finite sample distribution for $\rho$ in the neighbourhood of unity is very close to the asymptotic unit root case, under the assumption that the error term $\epsilon_t$ is normally distributed with finite variance.  

Statistical inference for M-estimators of possibly nonstationary time series models (local-to-unit root) is a nonstandard problem due to the presence of nuisance parameters in the limiting distributions of test statistics. However, indeed one of the  advantages of M-estimators is that are considered to be robust to outliers since they have a bounded influence function (see, \cite{abadir2000quantiles}). Considering now specifically the case of unit root such that $| \rho | = 1$, the asymptotic distribution of the t-statistic denoted by $\mathcal{T}_n ( \hat{\rho} )$ can be represented by functionals of Wiener processes (see, \cite{dickey1979distribution} and \cite{Buchmann2007asymptotic}). Thus, the asymptotic distribution of the $t-$statistic based on $M-$estimators, denoted by $\mathcal{T}_{\psi ( \uptau ) } ( \hat{\rho} )$ depends on the nuisance parameter $\delta$, that is, the correlation between the innovations $\left\{ \epsilon_t \right\}$ and the pseudo-score function $\psi( \epsilon_t )$ that is employed to define the $M-$estimator. On the other hand, $t-$statistics based on $M-$estimators lead to a reduction in asymptotic MSE relative to LSE for local alternatives to the unit-root null hypothesis. 

The t-statistic for the null hypothesis $\mathbb{H}_0: \rho = \rho_0$ is given by 
\begin{align}
t_{\psi} = \frac{ \displaystyle  \bigg( \hat{\rho}_n (\uptau) - \rho_n (\uptau) \bigg) }{ \displaystyle \left\{  \left( n^{-1} \sum_{t=1}^n \psi_{\uptau} \big( y_t - \hat{\rho}_n ( \uptau ) y_{t-1}   \big)^2  \right) \bigg/  \left( n^{-1} \sum_{t=1}^n \psi^{\prime}_{\uptau} \big( y_t - \hat{\rho}_n ( \uptau ) y_{t-1}   \big)  \right)^2 \right\}^{1/2} }.
\end{align}
where $\psi^{\prime}_{\uptau} (.)$ denotes the first derivative of the function $\psi_{\uptau} (.)$.
\color{black}

\newpage

\section{Asymptotic theory for Quantile Predictive Regression} 

In this section we unify the theory for the quantile predictive regression model while presenting the corresponding asymptotic theory for the quantile autoregressive process we introduced in the previous section with our novel pruned-based endogenous instrumentation approach.

\subsection{Model specification and assumptions}

Consider the following predictive regression model 
\begin{align}
y_t &= \alpha + \beta_n x_{t-1} + \varepsilon_t    
\\
x_t &= \mu + \phi_n x_{t-1} + u_t
\end{align}

Various studies in the literature consider inference methodologies for nearly-integrated processes. Specifically for the conditional quantile functional form \cite{lee2016predictive} propose a framework for inference in quantile predictive regression models using the IVX instrumentation of \cite{PM2009econometric}. However, robust inference in the near-explosive region for both the quantile autoregression and the predictive regression has not been examined previously in the literature.

In particular, OLS-based inference on $\beta$ for nearly-explosive processes and specifically purely-explosive suffers from the same problem as OLS-based inference on $\rho_n$, with standard inference applying only under $\textit{i.i.d}$ Gaussian innovations $\varepsilon_t$. Therefore, the inference procedure proposed in this paper for the quantile-dependent coefficient of $\beta$, in the quantile predictive regression model can accommodate regressons with time series properties along the entire spectrum of autoregressive procceses. Although our setting considers a model with only one regressor, therefore further theory is needed in the case of multiple regressors that correspond to either the same or different persistence class. Specifically, one can establish its asymptotic validity uniformly over the autoregressive regime and regardless of the distributional assumptions of the innovations $\varepsilon_t$ and $u_t$. Robust and uniform inference in quantile predictive regression models has been studied by \cite{maynard2023inference}, \cite{lee2016predictive}, \cite{fan2019predictive}, \cite{cai2022new} and more recently \cite{liu2023unified}.

The main idea with the proposed endogenously generated instrumentation is to decompose the the two mutually disjoint parameter space regions and then develop an asymptotic theory for both these regions. In other words, the main intuition behind this is that although the dependence structure of innovation sequences is not affected by considering two different regions of the parameter space, it does however have an effect on the asymptotic behaviour of the corresponding estimators which bridge the gap in both the explosive as well as the near-nonstationary cases. The proposed instrumentation method is based on the framework of \cite{Magdalinos2022uniform} and combines the nearly stable with the nearly explosive processes (near stationary/near explosive), thereby implying an adaptive and uniform inference technique for quantile autoregressions and quantile predictive regression models regardless of the peristence properties of regressors or whether the autoregressive and predictive equations include a model intercept.

\newpage

\subsection{Pruned-based endogenous instrumentation approach}

In this paper we propose a novel pruned-based endogenous instrumentation approach based on the original instrumentation methodology proposed in the paper of \cite{PM2009econometric} (see, also \cite{kostakis2015Robust}). The procedure we propose is similar to the IV instrument proposed in the recent paper of \textcolor{blue}{Magdalinos and Petrova (2022)}, although our motivation is to unify inference in autoregressions when considering a conditional quantile functional form, that is, is applied specifically to quantile autoregressive processes. We call our novel estimator IVX-P which can be employed in either linear or quantile conditional functional forms with univariate or multivariate regressors. In the literature such estimators mainly were concerned with Stein-type estimators, however our study is the first to consider an estimator which is obtained using properties of the underline stochastic processes both within the admissible parameter space as well as outside the usual parameter space.

Thus, the proposed pruned-based instrumental estimator implies a data driven instrument selection such that $\textcolor{blue}{ \mathcal{B}_n } = \mathbf{1} \left\{  n \left( \hat{\vartheta}_n^{ols} - 1 \right) \leq 0 \right\} $. More specifically, the chosen $\rho_{nz}$ is such that 
\begin{align}
\textcolor{violet}{ \vartheta_{nz} } 
&=
\varphi_{1n} \mathbf{1} \big\{ \textcolor{blue}{ \mathcal{B}_n } \big\} + \varphi_{2n} \mathbf{1} \big\{ \textcolor{red}{ \mathcal{B}_n^{c} } \big\}
\\
\varphi_{1n} &= \left( 1 - \frac{1}{ \kappa_n } \right), \ \ \ \varphi_{2n} = \left( 1 + \frac{1}{ \kappa_n } \right)
\end{align} 
where $\kappa_{1n}, \kappa_{2n} \to \infty$ with $\kappa_{1n} / n \to 0$ and $\kappa_{2n} / n \to 0$. Then, the combined instrument process is  
\begin{align}
z_t = \vartheta_{nz} z_{t-1} + \tilde{u}_t, \ \ \text{where} \ \ \tilde{u}_t = \Delta x_t \mathbf{1} \big\{ \textcolor{blue}{ \mathcal{B}_n } \big\} + \hat{u}_t \mathbf{1} \big\{ \textcolor{red}{ \mathcal{B}_n^{c} } \big\}
\end{align}
In other words, we have an orthogonal decomposition such that $z_t = z_{1t} \mathbf{1} \big\{ \textcolor{blue}{ \mathcal{B}_n } \big\} + z_{2t} \mathbf{1} \big\{ \textcolor{red}{ \mathcal{B}_n^{c} } \big\}$ where
\begin{align}
z_{1t} &= \textcolor{blue}{ \varphi_{1n} } z_{1t-1} + \textcolor{blue}{ \Delta x_t }
\\
z_{2t} &= \textcolor{red}{ \varphi_{2n} } z_{2t-1} + \textcolor{red}{ \hat{u}_t }
\end{align} 
Therefore, for all cases it holds that 
\begin{align}
n \left( \hat{\vartheta}_n^{ols} - 1 \right) 
&= 
\textcolor{blue}{ c }  \textcolor{red}{ \frac{n}{k_n} }  \left( 1 + \varepsilon_n \right), \ \ \varepsilon_n \overset{ p }{ \to } 0 
\nonumber
\\
&\overset{ p }{ \to } \mathsf{sign} ( \textcolor{blue}{ c }  ) \times \infty 
\end{align}
where $\tilde{z}_t^{IVX-P}$ represents the IVX pruned-based estimator. 
A key result that we are aiming to illustrate is the Asymptotic Mixed Gaussianity (AMG) property of the IVX-P estimator.

\newpage

\subsubsection{Instrument Construction} 

Successful instrumentation based on a combined near-stationary/near-explosive process requires statistical information separating the near-stationary autoregressive class from the near-explosive class asymptotically. In other words, the advantage of the inference procedure proposed in this paper over existing procedures is that it is valid for any $\rho_n \to \rho \in ( 0, + \infty )$, which includes all three parameter regions of interest of empirical interest.   
\color{black} Practical implementation of our instrumentation procedure requires a choice for $\varphi_{1n}$ and $\varphi_{2n}$. Specifically, choosing $\left( \varphi_{1n} \right)_{ n \in \mathbb{N} }$ and $\left( \varphi_{2n} \right)_{ n \in \mathbb{N} }$ with $n \left( \varphi_{1n} - 1 \right) \to \infty$ as $n \to \infty$ and $\varphi_{2n} \to 1$ with  $n \left( \varphi_{1n} - 1 \right) \to + \infty$.
\begin{align}
\varphi_{1n} = \left(  1 - \frac{1}{ n^{ \gamma_1} } \right) \ \ \text{and} \ \ \varphi_{2n} = \left(  1 + \frac{1}{ n^{ \gamma_2 } } \right)
\end{align} 
reduces to the problem of selecting the values for $\gamma_1$ and $\gamma_2$. 

We construct the instrumentation procedure based on a min-max optimality guarantee. In other words, we have that $\tilde{z}_{1t}$ can be asymptotically approximated by a near-stationary process such that  
\begin{align}
\tilde{z}_{1t} = \varphi_{1n} z_{1t-1} + u_t = \sum_{j=1}^t \varphi_{1n}^{t-j} u_j.   \end{align}
In particular, when $\rho_n$ is closer to 1 that $\varphi_{1n}$ then $\tilde{z}_{1t}$ reduces asymptotically to the original process $x_t$. Furthermore, the instrument $\tilde{z}_{2t}$ is always approximated by a mildly explosive process such that 
\begin{align}
z_{2t} = \varphi_{2n} z_{2t-1} + u_t = \sum_{j=1}^t \varphi_{2n}^{t-j} u_j    
\end{align}
Therefore, from the above decomposition we see that sample moments involving the near-stationary instrument $\tilde{z}_{1t}$ will contribute asymptotically when the original process $x_t$ belongs to the persistence classes \textbf{P.1}-\textbf{P.2} (i.e., near stationary and near-nonstationary), whearas sample moments involving the mildly-explosive instrument $\tilde{z}_{2t}$ will make an asymptotic contribution for autoregressions that belong to persistence classes  \textbf{P.2}-\textbf{P.3}. 

Moreover, under Assumption 4 we denote the autocovariance function and long-run variance of $\left( u_t \right)$ by $\gamma_u( . )$ and $\omega^2 = \sum_{ k = - \infty }^{ + \infty } \gamma_u ( k ) = C(1)^2 \sigma^2$ respectively and let 
\begin{align}
\gamma_n = \sum_{ k = 1 }^{ + \infty } \rho_n^{k-1} \gamma_u (k) \ \ \ \text{and} \ \ \ \Gamma =  \sum_{ k = 1 }^{ + \infty } \rho^{k-1} \gamma_u (k).     
\end{align}
Thus, we have that $\rho_N \to \rho$ and $\Gamma = \underset{ n \to \infty }{ \text{lim} } \Gamma_n$ exists by the dominated convergence theorem since $\sum_{ k = 1 }^{ \infty } \left| \gamma_u (k) \right| < + \infty$. Notice that when $\rho = 1, \Gamma = \sum_{k=1}^{+\infty} \gamma_u (k)$ is the one-sided long-run covariance of $\left(  u_t \right)_{ t \in \mathbb{N} }$. Denote with $W(t)$ to be the standard Brownian motion on $[0,1]$ and $B(t) = \omega W(t)$ and define the Ornstein-Uhlenbeck processes below
\begin{align}
W_c(t) = \int_0^t e^{c(t-s)} dW(s) \ \ \ \text{and} \ \ \ J_c(t) = \int_0^t e^{c(t-s)} dB(s)    
\end{align}

\newpage

\section{Conclusion}
\label{Section6}

In this paper we consider the asymptotic theory for moderate deviation from the unit boundary in quantile autoregressive and quantile predictive regression models. Using the moderate deviation principles we unify the asymptotic theory with a modified endogenous instrumentation procedure, without inducing limiting distribution discontinuities at certain regions of the parameter space. Specifically, in this study we verify the limit results obtained by \cite{Phillips2007limit} in the case of the linear autoregressive time series model. In particular, for both the case of near-stationary and near-explosive roots we establish the asymptotic theory of the quantile-dependent estimator which converges into a nuisance-parameter free limiting distribution.  

An extension of our framework in the regions which unifies all cases such as being in the unstable region with nearly stable or unstable processes such as the explosive and pure explosive processes, is an aspect of ongoing research that the author is actively undertaking. Further research aspects worth mentioning include the investigation of the asymptotic behaviour of quantile autoregressive models when a structural break occurs at an unknown break-point location. A relevant study using moderate deviations principles when testing for structural breaks include the framework proposed by \cite{xu2018limit} as is presented by \cite{katsouris2023structural}.

\newpage 

\appendix

\section{Moderate Deviations in Mildly Integrated and Explosive Cases}

Generally, it is found that quantile estimates are more robust than the ordinary least squares estimate when the underlying series is heavy-tailed. Strictly speaking, the moderate deviations from unity process is a triangular array such that $\left\{  y_{nt} : 1 \leq t \leq n  \right\}$. Notice that we consider the quantile estimation in the present paper and provide the asymptotic distribution for the quantile estimate.  Following the framework presented by \cite{wang2022asymptotics}. In particular, consider the first-order autoregressive model as below
\begin{align}
x_t = \rho_n x_{t-1} + u_t, \ \ \ t = 1,...,n, \ \ \rho_n = \left( 1 + \frac{c}{k_n} \right),   
\end{align}
We consider the truncated second moment $\ell (x) = \mathbb{E} \left[ u_1^2 \mathbf{1} \left\{ | u_1 | \leq x \right\} \right]$ is a slowly varying function of $x$ at $\infty$, which implies that $\mathsf{lim}_{ x \to \infty } \ell ( tx ) / \ell (x) = 1$ for all $t > 0$.  

\medskip

\underline{ \textbf{Mildly Integrated Case:} }

The following vector weak convergence holds:
\begin{align}
\left( \frac{1}{ \sqrt{n} } \sum_{t=1}^n \psi_{\uptau} \left( u_{t \uptau } \right), \frac{1}{  \sqrt{ n k_n \ell ( \eta_n ) } }  \sum_{t=1}^n x_{t-1} \psi_{\uptau} \left( u_{t \uptau } \right)   \right)  \Rightarrow \left( S ( \uptau ), T ( \uptau ) \right),  
\end{align}
where 
\begin{align}
S (\tau) \sim \mathcal{N} ( 0, \tau ( 1 - \tau ) ), \ \ \ \text{and} \ \  T ( t ) \sim \mathcal{N} \left( 0,  - \frac{ \tau (1 - \tau) }{ 2c }  \right) 
\end{align}
are independent random variables.     

\begin{proof}
Therefore, it suffices to show that for any $a, b \in \mathbb{R}$, 
\begin{align*}
\frac{ \textcolor{red}{a} }{ \sqrt{n} } \sum_{t=1}^n \psi_{\uptau} \left( u_{t \uptau } \right) + \frac{ \textcolor{red}{b} }{ \sqrt{ n k_n \ell( \eta_n ) } } \sum_{t=1}^n x_{t-1} \psi_{\uptau} \left( u_{t \uptau } \right)  
&\Rightarrow 
\textcolor{red}{a} S( \uptau ) + \textcolor{red}{b} T( \uptau ) 
\\
&= \mathcal{N} \left( 0,  \textcolor{red}{a^2} \uptau ( 1 - \uptau) - \frac{  \textcolor{red}{b^2} }{ 2c }  \uptau ( 1 - \uptau)  \right).
\end{align*}
Therefore, we denote with 
\begin{align}
\sum_{t=1}^n \varepsilon_{nt} :=  \frac{ \textcolor{red}{a} }{ \sqrt{n} } \sum_{t=1}^n \psi_{\uptau} \left( u_{t \uptau } \right) + \frac{ \textcolor{red}{b} }{ \sqrt{ n k_n \ell( \eta_n ) } } \sum_{t=1}^n x_{t-1} \psi_{\uptau} \left( u_{t \uptau } \right)    
\end{align}
where 
\begin{align}
\varepsilon_{nt} :=  \frac{ \textcolor{red}{a} }{ \sqrt{n} } \psi_{\uptau} \left( u_{t \uptau } \right) + \frac{ \textcolor{red}{b} }{ \sqrt{ n k_n \ell( \eta_n ) } } x_{t-1} \psi_{\uptau} \left( u_{t \uptau } \right), \ \ \ 1 \leq t \leq n.      
\end{align}

\newpage 

Then $\varepsilon_{nt}$ is a martingale difference with respect to the filtration $\mathcal{F}_{nt} = \sigma \left( y_0, u_1,..., u_t \right)$ and therefore the conditional variance of the martingale array $\sum_{t=1}^n \varepsilon_{nt}$ is given by 
\begin{align*}
\sum_{t=1}^n \mathbb{E} \left[  \varepsilon^2_{nt}  | \mathcal{F}_{nt-1} \right] 
&=  
\sum_{t=1}^n  \left(  \frac{ \textcolor{red}{a} }{ \sqrt{n} } \uptau ( 1 - \uptau ) + \frac{ \textcolor{red}{b} }{ \sqrt{ n k_n \ell( \eta_n ) } } x_{t-1}^2 \uptau ( 1 - \uptau )  +   \frac{ \textcolor{red}{2ab} }{ n \sqrt{ k_n \ell( \eta_n ) } } x_{t-1} \uptau ( 1 - \uptau)  \right)  
\\
&\overset{ p }{ \to } \textcolor{red}{ a^2 } \uptau ( 1 - \uptau) - \frac{ \textcolor{red}{ b^2 } }{ 2c } \uptau ( 1 - \uptau),   
\end{align*}
which holds since $\mathbb{E} \left[ \psi^2_{\uptau} \left( u_{t \uptau } \right) \right] = \uptau ( 1 - \uptau)$ and $\frac{ 1 }{ n k_n \ell( \eta_n )  } \sum_{t=1}^n x_{t-1}^2 \overset{ p }{ \to } - \frac{1}{2c}$ and due to the fact that 
\begin{align}
\frac{ 1  }{ n \sqrt{ k_n \ell( \eta_n ) } } \sum_{t=1}^n x_{t-1} = o_p(1).
\end{align}
\end{proof}

\begin{proof}
\color{red}
\begin{align*}
\frac{ 1  }{ n \sqrt{ k_n \ell( \eta_n ) } } \sum_{t=1}^n x_{t-1} = o_p(1).
\end{align*}
\color{black}
By expanding the expression we obtain 
\begin{align*}
\frac{1}{ n \sqrt{ k_n \ell( \eta_n ) } } \sum_{t=1}^n x_{t-1}  
&= 
\frac{x_0 }{ n \sqrt{ k_n \ell( \eta_n ) } } \sum_{t=1}^n \rho_n^{t-1} +  \frac{1}{ n \sqrt{ k_n \ell( \eta_n ) } } \sum_{t=1}^n \sum_{j=1}^{ t - 1 } \rho_n^{ t - 1 -j } u_j
\\
&= K_{n1} + K_{n2}. 
\end{align*}
It is easy to show that $K_{n1} = o_p(1)$, therefore we need to show that $K_{n2} = o_p(1)$. Thus, by adopting the truncation approach, we rewrite the expression as below 
\begin{align}
K_{n2} = \frac{1}{ n \sqrt{ k_n \ell( \eta_n ) } } \sum_{t=1}^n \sum_{j=1}^{t-1} \rho_n^{ t - 1 - j } u_j^{(1)} +  \frac{1}{ n \sqrt{ k_n \ell( \eta_n ) } }  \sum_{t=1}^n \sum_{j=1}^{t-1} \rho_n^{ t - 1 - j } u_j^{(1)} u_j^{(2)} := I + II.   
\end{align}
We have that 
\begin{align*}
\mathbb{E} \left[ I^2 \right] 
&= 
\frac{1}{ n^2 k_n \ell ( \eta_n ) } \mathbb{E} \left[ \sum_{j=1}^{ n - 1 } \sum_{t=j+1}^{n} \rho_n^{ t - 1 - j } u_j^{(1)} \right]^2  
\\
&= 
\frac{1}{ n^2 k_n \ell ( \eta_n ) }  \sum_{j=1}^{ n-1 } \left( \sum_{t=j+1}^{n}  \rho_n^{ t - 1 - j }  \right)^2 \ell ( \eta_n ) \big( 1 + o(1) \big)
\\
&\leq 
\frac{A}{n^2 k_n} \sum_{j=1}^{n-1} \left(  \frac{ 1 - \rho_n^{n-j}  }{ 1 - \rho_n  } \right) \leq \frac{A}{ n^2 k_n } \frac{ n }{ \left( 1 - \rho_n  \right)^2 } = O \left(  \frac{k_n}{n}  \right) = o(1). 
\end{align*}

\newpage

and 
\begin{align*}
\mathbb{E} \left[ | II | \right] 
&\leq 
\frac{2}{ n \sqrt{ k_n \ell ( \eta_n )} } . o \left( \frac{ \ell( \eta_n ) }{ \eta_n } \right) . \sum_{t=1}^n \sum_{ j=1 }^{t-1} \rho_n^{ t - 1 - j }    
\\
&\leq 
\frac{2}{ n \sqrt{ k_n \ell ( \eta_n )} } . o \left( \frac{ \ell( \eta_n ) }{ \eta_n } \right) . \frac{n}{ 1 - \rho_n }
\\
&= 
o \left( \frac{ \sqrt{ k_n \ell( \eta_n ) } }{ \eta_n } \right) = o \left( \sqrt{ \frac{k_n }{ n } } \right) = o(1),
\end{align*}
By using the martingale limit theorem, the proof is completed if the Linderberg condition
\begin{align}
\sum_{t=1}^n \mathbb{E} \left[ \varepsilon_{nt}^2 \mathbf{1} \left\{ | \varepsilon_{nt} | > \eta  \right\} \big| \mathcal{F}_{nt-1} \right] = o_p(1), \ \ \ \text{for any} \ \ \eta > 0.     
\end{align}
We have that, for all $\eta > 0$
\begin{align*}
\sum_{t=1}^n \mathbb{E} &\left[ \varepsilon_{nt}^2 \mathbf{1} \left\{ | \varepsilon_{nt} | > \eta  \right\} \big| \mathcal{F}_{nt-1} \right]      
\\
&= 
\sum_{t=1}^n \mathbb{E} \left[ \left( \frac{ \textcolor{red}{a} }{ \sqrt{n} } \psi_{\uptau} \left( u_{t \uptau } \right) + \frac{ \textcolor{red}{b} }{ \sqrt{ n k_n \ell( \eta_n ) } } x_{t-1} \psi_{\uptau} \left( u_{t \uptau } \right) \right)^2 \mathbf{1} \big\{ | \varepsilon_{nt} | > \eta \big\} \big| \mathcal{F}_{nt-1} \right]  
\\
&\leq 
2 \sum_{t=1}^n \mathbb{E} \left[ \left( \frac{ \textcolor{red}{a^2} }{ n } \psi_{\uptau}^2 \left( u_{t \uptau } \right) + \frac{ \textcolor{red}{b} }{ n k_n \ell( \eta_n ) } x_{t-1}^2 \psi_{\uptau}^2 \left( u_{t \uptau } \right) \right) \mathbf{1} \big\{ | \varepsilon_{nt} | > \eta \big\} \big| \mathcal{F}_{nt-1} \right]   
\\
&\leq 
\left( 2 \textcolor{red}{ a^2 } \uptau ( 1 - \uptau) + 2b^2 \uptau ( 1 - \uptau) . \frac{1}{ n k_n \ell( \eta_n ) } \sum_{t=1}^2 \sum_{t=1}^n y_{t-1}^2  \right) . \underset{ 1 \leq t \leq n }{ \mathsf{max} } \mathbb{E} \big[ \mathbf{1} \big\{ | \varepsilon_{nt} | > \eta \big\} \big| \mathcal{F}_{nt-1} \big]   
\end{align*}
Since it holds that $\frac{1}{ n k_n \ell( \eta_n ) } \sum_{t=1}^n x_{t-1}^2 = O_p(1)$, then it suffices to show that 
\begin{align}
\underset{ 1 \leq t  \leq n  }{ \mathsf{max} } \ \mathbb{E} \big[ \mathbf{1} \big\{ | \varepsilon_{nt} | > \eta \big\} \big| \mathcal{F}_{nt-1} \big] = o_p(1).     
\end{align}
Next, by applying the Chebyshev inequality we obtain that 
\begin{align}
\mathbb{E} \big[ \mathbf{1} \big\{ | \varepsilon_{nt} | > \eta \big\} \big| \mathcal{F}_{nt-1} \big] \leq \frac{ \mathbb{E} \left[ \varepsilon_{nt}^2 | \mathcal{F}_{nt-1} \right]  }{  \eta^2 } 
\leq \frac{  2 a^2 \uptau ( 1 - \uptau) }{ n \eta^2 }  + \frac{ 2 b^2 \uptau ( 1 - \uptau) }{ \eta^2 } . \frac{  x_{t-1}^2 }{ n k_n \ell(\eta_n)  }.  
\end{align}
Then the previous expression holds since $\mathsf{max}_{ 1 \leq t \leq n } \frac{ | x_{t-1} | }{ \sqrt{ n k_n \ell( \eta_n ) } } = o_p(1)$. 
\end{proof}

\newpage 

\paragraph{Proof of Theorem 3.1:}

\

Using the Taylor's expansion we have that:
\begin{align*}
\sum_{t=1}^n \mathbb{E} \big[ \xi_t ( \boldsymbol{\mathsf{v}}  ) \big]
&= 
\sum_{t=1}^n \int_0^{ \boldsymbol{\mathsf{v}}^{\top} \boldsymbol{D}_n^{-1} x_t  } \mathbb{E} \big[ \boldsymbol{1} \left\{ u_t \leq \beta ( \tau) + s   \right\} - \boldsymbol{1} \left\{ u_t \leq \beta (\tau) \right\}  \big] ds  \\  
&= 
\sum_{t=1}^n \int_0^{ \boldsymbol{\mathsf{v}}^{\top} \boldsymbol{D}_n^{-1} x_t  }  \big\{ F \big( \beta(\tau) + s \big) - F \big( \beta ( \tau ) \big) \big\} ds
\\
&= 
\sum_{t=1}^n \int_0^{ \boldsymbol{\mathsf{v}}^{\top} \boldsymbol{D}_n^{-1} x_t  }  \left\{  s \cdot f \big( \beta (\tau) \big) + \frac{1}{2} s^2 \cdot f^{\prime} (s^{*} ) \right\} ds
\\
&= \frac{  f \big( \beta (\tau) \big)  }{2} \cdot \boldsymbol{\mathsf{v}}^{\top} \left(  \sum_{t=1}^n \boldsymbol{D}_n^{-1} x_t x_t^{\top} \boldsymbol{D}_n^{-1}  \right) \boldsymbol{\mathsf{v}}    + \frac{1}{2} \sum_{t=1}^n \int_0^{ \boldsymbol{\mathsf{v}}^{\top} \boldsymbol{D}_n^{-1} x_t  }  s^2 \cdot f^{\prime} (s^{*} ) ds,
\end{align*}
where $s^{*} \in \big( \beta(\tau), \beta(\tau) + s \big)$. We have that 
\begin{align}
\sum_{t=1}^n \boldsymbol{D}_n^{-1} x_t x_t^{\prime} \boldsymbol{D}_n^{-1} 
=
\begin{pmatrix}
1 &  \displaystyle \frac{1}{ n \sqrt{ k_n \ell( \eta_n ) } } \sum_{t=1}^n x_{t-1} 
\\
 \displaystyle \frac{1}{ n \sqrt{ k_n \ell( \eta_n ) } } \sum_{t=1}^n x_{t-1}  &  \displaystyle \frac{1}{ n k_n \ell( \eta_n )  } \sum_{t=1}^n x_{t-1}^2  
\end{pmatrix}
\overset{ p }{  \to } 
\Sigma
= 
\begin{pmatrix}
1 & 0 
\\
0  & - \frac{1}{2c}
\end{pmatrix}.
\end{align}

Hence, we can prove that 
\begin{align}
\sum_{t=1}^n \mathbb{E} \big[ \xi_t ( \boldsymbol{\mathsf{v}}  ) | \mathcal{F}_{nt-1}  \big] = \frac{  f \big( \beta (\tau) \big)  }{2} \cdot \boldsymbol{\mathsf{v}}^{\top} \boldsymbol{\Sigma}  \boldsymbol{\mathsf{v}} + o_p(1). 
\end{align}
Notice that $\sum_{t=1}^n \big(  \xi_t ( \boldsymbol{\mathsf{v}}  ) -  \mathbb{E} \big[ \xi_t ( \boldsymbol{\mathsf{v}}  ) | \mathcal{F}_{nt-1}  \big]  \big)$ is a martingale difference sequence and since it also holds that $\sum_{t=1}^n \mathbb{E} \big[ \xi^2_t ( \boldsymbol{\mathsf{v}}  ) | \mathcal{F}_{nt-1}  \big]$ then it follows that $\sum_{t=1}^n \big(  \xi_t ( \boldsymbol{\mathsf{v}}  ) -  \mathbb{E} \big[ \xi_t ( \boldsymbol{\mathsf{v}}  ) | \mathcal{F}_{nt-1}  \big]  \big) = o_p(1)$.

Therefore, it holds that 
\begin{align}
Z_n( \boldsymbol{v} ) := - \boldsymbol{\mathsf{v}}^{\top} R_n(\tau) +  \frac{  f \big( \beta (\tau) \big)  }{2} \cdot \boldsymbol{\mathsf{v}}^{\top} \boldsymbol{\Sigma}  \boldsymbol{\mathsf{v}} + o_p(1).    
\end{align}
Moreover, the following joint weakly convergence results for the two functionals:
\begin{align}
R_n(\tau) := \left( \frac{1}{\sqrt{n}} \sum_{t=1}^n \psi_{\tau} \big( u_{t \tau}  \big), \frac{1}{ \sqrt{ n k_n \ell ( \eta_n ) } } \sum_{t=1}^n y_{t-1} \psi_{\tau} \big( u_{t \tau}  \big) \right)^{\top} \Rightarrow \big(  \mathcal{S}( \tau ),  \mathcal{J}(\tau) \big)^{\top}.   
\end{align}
Thus, it holds that 
\begin{align}
Z_n( \boldsymbol{v} ) := - \boldsymbol{\mathsf{v}}^{\top} \cdot  \big(  \mathcal{S}( \tau ),  \mathcal{J}(\tau) \big)^{\top} +  \frac{  f \big( \beta (\tau) \big)  }{2} \cdot \boldsymbol{\mathsf{v}}^{\top} \boldsymbol{\Sigma}  \boldsymbol{\mathsf{v}}.    
\end{align}
Since $Z_n( \boldsymbol{v}  )$ has convex sample paths, it implies uniform convergence on compact sets.

\newpage 

\paragraph{Proof of Theorem 4.1:}

\

By the Cr\'amer-Wold principle, it suffices to show for any $a, b, c \in \mathbb{R}$, 
\begin{align}
\textcolor{red}{ \alpha } J_n ( \uptau ) + \textcolor{red}{ \beta } K_n( \uptau ) + \textcolor{red}{ \gamma } L_n \Rightarrow \textcolor{red}{ \alpha } J ( \uptau ) + \textcolor{red}{ \beta } K( \uptau ) + \textcolor{red}{ \gamma } L    
\end{align}
In particular, we rewrite 
\begin{align*}
\textcolor{red}{ \alpha } &J_n ( \uptau ) + \textcolor{red}{ \beta } K_n( \uptau ) + \textcolor{red}{ \gamma } L_n   \\
&=
\sum_{t=1}^n \left[ \left(  \frac{ \textcolor{red}{ \alpha }  }{ \sqrt{n} } + \frac{ \textcolor{red}{ \beta }  }{ \sqrt{k_n} } \rho_n^{ - (n-t) - 1 } \right) \psi_{\uptau} \left( u_{t \uptau } \right) + \frac{ \textcolor{red}{ \gamma }  }{ \sqrt{ k_n \ell ( \eta_n ) } }  \sum_{t=1}^n \rho_n^{-t} u_t^{(2)} \right] 
=: R_{n1} + R_{n2}.
\end{align*}
Thus, it holds that 
\begin{align}
\mathbb{E} \left| R_{n2} \right| \leq \frac{A}{ \sqrt{ k_n \ell(\eta_n) }  } . o \left( \frac{ \ell( \eta_n ) }{ \eta_n } \right) . \frac{1 - \rho_n^{-n} }{ \rho_n - 1 }   
\leq 
A c \sqrt{k_n} . o \left( \frac{ \sqrt{ \ell( \eta_n ) } }{ \eta_n } \right) \leq A c \sqrt{k_n} . o \left( \frac{1}{ \sqrt{n} } \right) = o(1),  
\end{align}
therefore, $R_{n2} = o_p(1)$. Furthermore, to show that 
\begin{align}
R_{n1} \overset{ d }{ \to } \mathcal{N} \left( 0,  \textcolor{red}{ \alpha^2} \uptau ( 1 - \uptau) + \frac{  \textcolor{red}{ \beta^2} }{ 2c } \uptau ( 1 - \uptau)  +  \frac{  \textcolor{red}{ \gamma^2 }  }{ 2c } \right)    
\end{align}
In particular, denote with $R_{n1} =: \sum_{t=1}^n \delta_{nt}$, where
\begin{align}
\delta_{nt} =  \left\{ \frac{ \textcolor{red}{ \alpha }  }{ \sqrt{n} } \psi_{\uptau} \left( u_{t \uptau } \right)  + \frac{ \textcolor{red}{  \beta }  }{ \sqrt{k_n} } \rho_n^{ - (n-t) - 1 }  \psi_{\uptau} \left( u_{t \uptau } \right)  + \frac{ \textcolor{red}{ \gamma } }{ \sqrt{ k_n \ell ( \eta_n ) } }  \sum_{t=1}^n \rho_n^{-t} u_t^{(1)} \right\}
\end{align}
Then $\left\{  \delta_{nt} \right\}$ are independent random variables with $\mathbb{E} \left(  \delta_{nt} \right) = 0$. Then, the variance of the \textit{m.d.s} is  
\begin{align*}
\mathbb{E} \left( \sum_{t=1}^n \delta_{nt} \right)^2 
&=
\mathbb{E} \left[ \sum_{t=1}^n \left( \frac{ \textcolor{red}{ \alpha }  }{ \sqrt{n} } \psi_{\uptau} \left( u_{t \uptau } \right) \right)^2 +  \sum_{t=1}^n \left( \frac{ \textcolor{red}{  \beta }  }{ \sqrt{k_n} } \rho_n^{ - (n-t) - 1 }  \psi_{\uptau} \left( u_{t \uptau } \right)  \right)^2  +  \frac{ \textcolor{red}{ \gamma^2 } }{ k_n \ell ( \eta_n ) }  \sum_{t=1}^n \rho_n^{-2t} u_t^{(1)^2}  \right]
\\
&\ \  
+ 2 \mathbb{E} \left[ \sum_{t=1}^n \left( \frac{ \textcolor{red}{ \alpha }  }{ \sqrt{n} }   + \frac{ \textcolor{red}{  \beta }  }{ \sqrt{k_n} } \rho_n^{ - (n-t) - 1 }  \right) \psi_{\uptau} \left( u_{t \uptau } \right) \frac{ \textcolor{red}{ \gamma } }{ \sqrt{ k_n \ell ( \eta_n ) } }  \sum_{t=1}^n \rho_n^{-2t} u_t^{(1)}  \right]
\\
&= 
\textcolor{red}{ \alpha^2 } \uptau ( 1 - \uptau) + \textcolor{red}{ \beta^2 } \uptau ( 1 - \uptau) . \frac{1}{k_n} \sum_{t=1}^n \rho_n^{ -2 (n-t) - 1} + \textcolor{red}{ \gamma^2 } . \frac{1}{ k_n } \sum_{t=1}^n \rho_n^{-2t} \left( 1 + o_p(1) \right)
\\
&\ \ \ +
2 \textcolor{red}{ \alpha } \textcolor{red}{ \beta } \uptau ( 1 - \uptau) . \frac{1}{ \sqrt{k_n} }  \sum_{t=1}^n \rho_n^{ - (n-t) - 1}  + \underbrace{ \mathbb{E} \left[ \frac{ 2 \textcolor{red}{ \alpha } \textcolor{red}{ \gamma } }{ \sqrt{ n k_n \ell ( \eta_n ) } } \sum_{t=1}^n \rho_n^{-t}  \psi_{\uptau} \left( u_{t \uptau } \right) u_t^{(1)}  \right] }_{ III }
\\
&\ \ \ +
\underbrace{ \mathbb{E} \left[ \frac{ 2 \textcolor{red}{ \beta } \textcolor{red}{ \gamma } }{ k_n \sqrt{ n  \ell ( \eta_n ) } } \sum_{t=1}^n \rho_n^{-( n + 1 )}  \psi_{\uptau} \left( u_{t \uptau } \right) u_t^{(1)} \right] }_{ IV }
\end{align*}

\newpage 

Notice that we have that 
\begin{align}
\frac{1}{ \sqrt{n k_n} } \sum_{t=1}^n \rho_n^{ - ( n - t) - 1} =  \frac{1}{ \sqrt{n k_n} } \frac{ 1 - \rho_n^{-n}  }{ \rho_n - 1 } \to 0. 
\end{align}
Moreover, by the triangular inequality it holds that 
\begin{align}
\mathbb{E} \left| \psi_{\uptau} \left( u_{t \uptau } \right) \frac{ u_t^{(1)} }{ \sqrt{\ell(\eta_n ) } } \right| \leq \sqrt{ \mathbb{E} \left[ \psi^2_{\uptau} \left( u_{t \uptau } \right)  \right] } . \sqrt{ \mathbb{E} \left[ \left( \frac{ u_t^{(1)} }{ \sqrt{\ell(\eta_n ) } }  \right)^2 \right] } = O(1),  
\end{align}
Therefore, it holds that 
\begin{align}
\left| III \right| \leq \mathbb{E} \left| \frac{ 2 \textcolor{red}{ \alpha } \textcolor{red}{ \gamma } }{ \sqrt{ n k_n  } } \sum_{t=1}^n \rho_n^{-t}  \psi_{\uptau} \left( u_{t \uptau } \right)  \frac{ u_t^{(1)}  }{ \sqrt{ \ell ( \eta_n ) } } \right| \leq O(1) . \frac{1}{ \sqrt{ n k_n  } } \sum_{t=1}^n \rho_n^{-t} = o(1),   
\\
\left| IV \right| \leq \mathbb{E} \left|  \frac{ 2 \textcolor{red}{ \beta } \textcolor{red}{ \gamma } }{  k_n } \sum_{t=1}^n \rho_n^{-(n+1)}  \psi_{\uptau} \left( u_{t \uptau } \right)  \frac{ u_t^{(1)}  }{ \sqrt{ \ell ( \eta_n ) } } \right| \leq O(1) . \frac{n}{ k_n } \rho_n^{-(n+1)} = o(1).   
\end{align}
Hence, it holds that
\begin{align*}
\mathbb{E} \left( \sum_{t=1}^n \delta_{nt} \right)^2 = \textcolor{red}{ \alpha^2} \uptau ( 1 - \uptau) + \frac{  \textcolor{red}{ \beta^2} }{ 2c } \uptau ( 1 - \uptau)  +  \frac{  \textcolor{red}{ \gamma^2 }  }{ 2c } + o_p(1)
\end{align*}
In other words, since we have shown that the conditional Lindeberg condition for the martingale difference array $\sum_{t=1}^n \delta_{nt}$ holds and we determine its conditional variance, then the proof of the theorem is completed.

\paragraph{Proof of Lemma 4.2:}

\

\begin{align*}
\frac{ \rho_n^{ -2n + 1 } }{ k_n \ell ( \eta_n ) } \sum_{t=1}^n x_{t-1} u_t
&= 
\frac{ \rho_n^{ -n + 1 } x_0  }{ \sqrt{ k_n \ell ( \eta_n ) } }  \frac{ 1  }{ \sqrt{ k_n \ell ( \eta_n ) } } \sum_{t=1}^n \rho_n^{ - (n - t) - 1  } u_t +  \frac{ \rho_n^{ -2n + 1 } x_0  }{  k_n \ell ( \eta_n ) }  \sum_{t=1}^n \left( \sum_{j=1}^{t-1} \rho_n^{t-1-j} u_j \right) u_t
\\
&= 
\frac{ \rho_n^{ -2n + 1 } x_0  }{  k_n \ell ( \eta_n ) }  \sum_{t=1}^n \left( \sum_{j=1}^{t-1} \rho_n^{t-1-j} u_j \right) u_t + o_p(1),
\end{align*}
since $x_0 = o_p ( \sqrt{k_n} )$ and $\rho_n^{-n} = o( k_n / n )$ and $\frac{1}{ \sqrt{k_n \ell(\eta_n) } } \sum_{t=1}^n \rho_n^{ - (n-t) - 1 } u_t \overset{ d }{ = } L_n \overset{ d }{ \to } L$. Therefore, it suffices to show that
\begin{align}
\frac{ \rho_n^{ -2n + 1 } }{ k_n \ell ( \eta_n ) }  \sum_{t=1}^n \left( \sum_{j=1}^{ t - 1 } \rho_n^{ t - 1 - j } u_j  \right) u_t = o_p(1). 
\end{align}

\newpage

Moreover, we rewrite the following expression 
\begin{align*}
\frac{ \rho_n^{ -2n + 1 } }{ k_n \ell ( \eta_n ) } \sum_{t=1}^n \left( \sum_{j=1}^{ t - 1 } \rho_n^{ t - 1 - j } u_j  \right) u_t 
&=  
\frac{ \rho_n^{ -2n + 1 } }{ k_n \ell ( \eta_n ) } \sum_{t=1}^n \left( \sum_{j=1}^{ t - 1 } \rho_n^{ t - 1 - j } u_j^{(1)} \right) u_t^{(1)} 
+
\frac{ \rho_n^{ -2n + 1 } }{ k_n \ell ( \eta_n ) } \sum_{t=1}^n \left( \sum_{j=1}^{ t - 1 } \rho_n^{ t - 1 - j } u_j^{(1)} \right) u_t^{(2)}
\\
&+ \frac{ \rho_n^{ -2n + 1 } }{ k_n \ell ( \eta_n ) } \sum_{t=1}^n \left( \sum_{j=1}^{ t - 1 } \rho_n^{ t - 1 - j } u_j^{(2)} \right) u_t^{(1)} + \frac{ \rho_n^{ -2n + 1 } }{ k_n \ell ( \eta_n ) } \sum_{t=1}^n \left( \sum_{j=1}^{ t - 1 } \rho_n^{ t - 1 - j } u_j^{(2)} \right) u_t^{(2)}
\\
&=: \mathcal{S}_{n1} + \mathcal{S}_{n2} + \mathcal{S}_{n3} + \mathcal{S}_{n4}. 
\end{align*}
Therefore by Lemma A.1 we can obtain the following result
\begin{align*}
\mathbb{E} \left[ \mathcal{S}_{n1}  \right]^2  
&= 
\mathbb{E} \left[ \frac{ \rho_n^{ -2n + 1 } }{ k_n \ell ( \eta_n ) } \sum_{t=1}^n \left( \sum_{j=1}^{ t - 1 } \rho_n^{ t - 1 - j } u_j^{(1)} \right) u_t^{(1)}   \right]^2
\\
&\leq 
\frac{ A \rho_n^{ - 4n + 2 }  }{ k_n^2 } \sum_{t=1}^n \sum_{j=1}^{ t - 1 } \rho_n^{ 2 ( t - 1 - j ) }
\\
&=
\frac{ A \rho_n^{ - 4n + 2 }  }{ k_n^2 \textcolor{blue}{ \left( \rho_n^2 - 1 \right)} } \left[  \sum_{t=1}^n \rho_n^{ 2 ( t - 1) }  - n \right] = O \left( \rho_n^{ -2n _ 2 } \right) = o(1), 
\end{align*}
which shows that $\mathcal{S}_{n1} = o_p(1)$. Moreover, it holds that 
\begin{align*}
\mathcal{S}_{n2} = \textcolor{red}{ \frac{ \rho_n^{ - 2n + 1 }  }{ k_n \ell ( \eta_n ) } }  \sum_{t=1}^n \left( \sum_{j=1}^{ t - 1 } \rho_n^{ t - 1 - j } u_j^{(1)} \right) u_t \mathbf{1} \big\{ \left| u_t \right| > \eta_n \big\} + o \left( \frac{ \ell ( \eta_n ) }{ \eta_n } \right) . \textcolor{red}{ \frac{ \rho_n^{ - 2n + 1 }  }{ k_n \ell ( \eta_n ) } } \sum_{t=1}^n \sum_{j=1}^{ t - 1 } \rho_n^{t - 1 - j } u_j^{(1)} 
=: 
V + VI.
\end{align*}
Therefore, using the Cauchy-Schwarz inequality, we obtain the following result
\begin{align*}
\left| V \right| 
&\leq 
\textcolor{red}{ \frac{ \rho_n^{ - 2n + 1 }  }{ k_n \ell ( \eta_n ) } }  \sqrt{  \sum_{t=1}^n u_t^2 } \sqrt{ \sum_{t=1}^n \left( \sum_{j=1}^{t-1} \rho_n^{t-1-j} u_j^{(1)} \right)^2  \mathbf{1} \big\{ \left| u_t \right| > \eta_t \big\}  }    
\\
&= 
\sqrt{ \sum_{t=1}^n u_t^2 / n \ell ( \eta_n ) }  \frac{ \sqrt{n} \rho_n^{2n + 1 }  }{ k_n \sqrt{ \ell ( \eta_n ) } }  \sqrt{ \sum_{t=1}^n \left( \sum_{j=1}^{t-1} \rho_n^{t-1-j} u_j^{(1)} \right)^2  \mathbf{1} \big\{ \left| u_t \right| > \eta_t \big\}  }.   
\end{align*}
Then, a simple application of Jenesen's inequality yields the following result
\begin{align*}
\mathbb{E} \left[ \frac{ \sqrt{n} \rho_n^{- 2n + 1 }  }{ k_n \sqrt{ \ell ( \eta_n ) } }  \sqrt{ \sum_{t=1}^n \left( \sum_{j=1}^{t-1} \rho_n^{t-1-j} u_j^{(1)} \right)^2  \mathbf{1} \big\{ \left| u_t \right| > \eta_t \big\}  } \right]  
&\leq 
\frac{ \sqrt{n} \rho_n^{- 2n + 1 }  }{ k_n \sqrt{ \ell ( \eta_n ) } }  \sqrt{ \sum_{t=1}^n \sum_{j=1}^{t-1} \rho_n^{ 2(t-1-j) } \ell( \eta_n ) \big[ 1 + o(1)  \big] \mathbb{P} \big( \left| u_t \right| > \eta_t \big)  } 
\\
&\leq 
\frac{ \sqrt{n} \rho_n^{- 2n + 1 }  }{ k_n \sqrt{ \ell ( \eta_n ) } }  \sqrt{ k_n^2 \ell( \eta_n ) . o \left( \frac{1}{n}  \right)  } = o(1),
\end{align*}
which together with $\sum_{t=1}^n u_t^2 / n \ell ( \eta_n ) \overset{ p }{ \to } 1$ implies that $V = o_p(1)$.

\newpage 

Therefore, Lemma A.1 yields the following result
\begin{align*}
\mathbb{E} \left| T_{n1} \right| 
&\leq 
\frac{ \rho_n^{ - n - 1 }  }{ k_n } \sum_{t=1}^n \mathbb{E} \left|  \frac{ u_t^{(1)} }{  \sqrt{ \ell ( \eta_n ) }  } \psi_{\uptau} \left( u_{t \uptau} \right) \right|  = O \left(  \rho_n^{ - n - 1 } \frac{n}{k_n}  \right) = o(1), 
\\
\mathbb{E} \left| T_{n2} \right| &\leq \frac{ \rho_n^{ - n - 1 }  }{ k_n  \sqrt{ \ell ( \eta_n ) } } \sum_{t=1}^n \mathbb{E} \left| u_t^{(2)} \psi_{\uptau} \left( u_{t \uptau} \right)  \right| 
\\
&\leq 
\frac{ \rho_n^{ - n - 1 }  }{ k_n  \sqrt{ \ell ( \eta_n ) } } . nA . o \left( \frac{ \ell ( \eta_n ) }{ \eta_n }  \right) = \frac{ \rho_n^{ - n - 1 }  }{ k_n  } . o \left( \frac{ \sqrt{ \ell ( \eta_n ) } }{ \eta_n }  \right) = o(1).
\end{align*}
Furthermore, we have that 
\begin{align*}
\mathbb{E} \left( T_{n3}^2 \right) 
&= 
\frac{ \rho_n^{ - 2n }  }{ k_n \ell ( \eta_n )   } . \uptau (1 - \uptau)   \sum_{t=1}^n \sum_{j=t+1}^n \rho_n^{2 ( t - j - 1) }  \ell ( \eta_n ) \big[ 1 + o(1)  \big]  
\\
&\leq A. \frac{ \uptau (1 - \uptau) \rho_n^{ -2n - n } }{ k_n^2  } . \frac{n}{ \rho_n^2 - 1} = O \left(  \frac{ \rho_n^{ -2n } n  }{ k_n }  \right) = o(1).
\end{align*}
Therefore, the quantities $T_{n1}, T_{n2}, T_{n3}$ are all $o_p(1)$. Next, we also show that $T_{n4} = o_p(1)$. In particular, it holds that 
\begin{align*}
T_{n4} =  \frac{ \rho_n^{-n} }{ k_n \sqrt{ \ell( \eta_n ) } } \sum_{t=1}^n \left( \sum_{j=t+1}^n \rho_n^{t-1-j} u_j \mathbf{1} \big\{ | u_j | > \eta_n \big\}   \right) \psi_{\tau} ( u_{t \tau} ) + o \left( \frac{ \ell ( \eta_n ) }{ \eta_n }  \right) . \frac{ \rho_n^{-n} }{ k_n \sqrt{ \ell( \eta_n ) } } \sum_{t=1}^n \sum_{j=t+1}^n \rho_n^{t-1-j} \psi_{\tau} ( u_{t \tau} ).  
\end{align*}

\begin{align*}
\frac{ \rho_n^{-2n} }{  k_n^2 \ell ( \eta_n  ) } \sum_{t=1}^n x_{t-1}^2
&= 
\frac{1}{ k_n^2 \ell ( \eta_n ) \left( \rho_n^2 - 1 \right) } \left[ \rho_n^{-2n} \left( x_n^2 - x_0^2 \right) - 2 \rho_n^{-2n+1} \sum_{t=1}^n x_{t-1} u_t - \rho_n^{-2n} \sum_{t=1}^n u_t^2 \right]
\\
&= 
\frac{1}{ k_n \left( \rho_n^2 - 1 \right) } \left[ \frac{ \rho_n^{-2n} }{ k_n \ell ( \eta_n ) } x_n^2 - \frac{ 2 \rho_n^{-2n+1} }{ k_n \ell ( \eta_n ) } \sum_{t=1}^n x_{t-1} u_t -  \frac{ \rho_n^{-2n}  }{ k_n \ell ( \eta_n ) } \sum_{t=1}^n u_t^2 \right] + o_p ( \rho_n^{ -2n } ).
\end{align*}
Since $k_n \left( \rho_n^2 - 1 \right) \to 2c$, $\rho_n^{-n} = o \left( \frac{k_n}{n} \right)$, 
\begin{align}
\frac{ \rho_n^{ - 2n } }{ k_n \ell \left( \eta_n \right) } \sum_{t=1}^n  u_t^2 = O_p \left( \frac{ n }{ k_n } \rho_n^{ - 2n }  \right) = o_p(1).  
\end{align}
Note that we also have that $\frac{ \rho_n^{ - 2n + 1 } }{ k_n \ell \left( \eta_n \right) } \sum_{t=1}^n  x_{t-1} u_t = o_p(1)$, then we obtain that 
\begin{align*}
\frac{ \rho_n^{ - 2n } }{ k_n^2 \ell \left( \eta_n \right) } \sum_{t=1}^n  x^2_{t-1} 
&= 
\frac{1}{ k_n \left( \rho_n^2 - 1 \right)} \left( \frac{ \rho_n^{-n} }{ \sqrt{k_n \ell ( \eta_n ) } } \right)^2 + o_p(1)      
\\
&= 
\frac{1}{ k_n \left( \rho_n^2 - 1 \right)} \left\{  \frac{x_0}{ \sqrt{ k_n \ell \left( \eta_n \right)}  }  +  \frac{1}{ \sqrt{ k_n \ell \left( \eta_n \right)}  }   \sum_{t=1}^n \rho_n^{ -t } u_t  \right\}^2 + o_p(1)   
= 
\frac{1}{ 2 c } L_n^2 +  o_p(1).    
\end{align*}

\newpage 

Moreover, by Lemma we obtain that 
\begin{align*}
&\frac{ \rho_n^{-n} }{ k_n \sqrt{ \ell \left( \eta_n \right) } } \sum_{t=1}^n x_{t-1} \psi_{\tau} \left( u_{t \tau} \right)     
\\
&\ = 
\frac{ x_0 }{ \sqrt{ k_n \ell \left( \eta_n \right) } } . \frac{1}{ \sqrt{k_n} } \sum_{t=1}^n \rho_n^{-( n - t - 1)} \psi_{\tau} \left( u_{t \tau} \right)  + \frac{ \rho_n^{-n} }{ k_n \sqrt{ \ell \left( \eta_n \right) } } \sum_{t=1}^n \left( \sum_{j=1}^{t-1} \rho_n^{t-1-j} u_j \right) \psi_{\tau} \left( u_{t \tau} \right) 
\\
&=
\frac{ x_0 }{ \sqrt{ k_n \ell \left( \eta_n \right) } } . K_n ( \tau )  + \frac{ \rho_n^{-n} }{ k_n \sqrt{ \ell \left( \eta_n \right) } } \sum_{t=1}^n \left( \sum_{j=1}^{t-1} \rho_n^{t-1-j} u_j \right) \psi_{\tau} \left( u_{t \tau} \right)
\\
&=
\frac{ \rho_n^{-n} }{ k_n \sqrt{ \ell \left( \eta_n \right) } } \sum_{t=1}^n \left( \sum_{j=1}^{t-1} \rho_n^{t-1-j} u_j \right) \psi_{\tau} \left( u_{t \tau} \right) + o_p(1)
\\
&=
\frac{ \rho_n^{-n} }{ k_n \sqrt{ \ell \left( \eta_n \right) } } \sum_{t=1}^n \left( \sum_{j=1}^{n} \rho_n^{t-1-j} u_j \right) \psi_{\tau} \left( u_{t \tau} \right) - \frac{ \rho_n^{-n} }{ k_n \sqrt{ \ell \left( \eta_n \right) } } \sum_{t=1}^n \left( \sum_{j=t}^{n} \rho_n^{t-1-j} u_j \right) \psi_{\tau} \left( u_{t \tau} \right) + o_p(1)
\\
&=
\left( \frac{1}{ \sqrt{k_n} } \sum_{t=1}^n \rho_n^{- (n - t + 1)} \psi_{\tau} \left( u_{t \tau} \right) \right). \left( \frac{1}{ \sqrt{ k_n \ell ( \eta_n ) } } \sum_{j=1}^n \rho_n^{-j} u_j \right) + o_p(1),
\\
&=
K_n ( \tau ) L_n + o_p(1). 
\end{align*}
Furthermore, for $\forall \ x > 0$, we have that 
\begin{align*}
\mathbb{P} &\left( \underset{ 1 \leq t \leq n  }{ \mathsf{max} } \left| \frac{ \rho_n^{-n} }{ k_n \sqrt{ \ell \left( \eta_n \right) } } x_{t-1}  \right| > x    \right)    
\\
&= 
\mathbb{P} \left( \underset{ 1 \leq t \leq n  }{ \mathsf{max} } \left| \frac{ \rho_n^{-n} }{ \sqrt{ \ell \left( \eta_n \right) } } \left( \rho_n^{t-1} x_0 + \sum_{j=1}^{t-1} \rho_n^{t - 1 - j} u_j  \right) \right| >  x k_n \right)
\\
&\leq 
\mathbb{P} \left( \underset{ 1 \leq t \leq n  }{ \mathsf{max} } \rho_n^{t-n-1} \left|  \frac{x_0}{ \sqrt{ \ell( \eta_n ) } } \right| >  \frac{ k_n x}{ 3 } \right) 
+
\mathbb{P} \left( \underset{ 1 \leq t \leq n  }{ \mathsf{max} } \rho_n^{t-n-1} \left| \sum_{j=1}^{t-1} \rho_n^{-j} \frac{  u_j^{(1)} }{ \sqrt{ \ell( \eta_n ) } } \right| >  \frac{ k_n x}{ 3 } \right)   
\\
&\ \ \ + \mathbb{P} \left(  \underset{ 1 \leq t \leq n  }{ \mathsf{max} }  \rho_n^{t-n-1} \left| \sum_{j=1}^{t-1} \rho_n^{-j} \frac{  u_j^{(2)} }{ \sqrt{ \ell( \eta_n ) } } \right| >  \frac{ k_n x}{ 3 }  \right)
\\
&\leq
\mathbb{P} \left( \underset{ 1 \leq t \leq n  }{ \mathsf{max} } \left|  \frac{x_0}{ \sqrt{ \ell( \eta_n ) } } \right| >  \frac{ k_n x}{ 3 } \right) 
+ \mathbb{P} \left( \underset{ 1 \leq t \leq n  }{ \mathsf{max} } \left| \sum_{j=1}^{t-1} \rho_n^{-j} \frac{  u_j^{(1)} }{ \sqrt{ \ell( \eta_n ) } } \right| >  \frac{ k_n x}{ 3 } \right)  + \mathbb{P} \left(  \underset{ 1 \leq t \leq n  }{ \mathsf{max} } \left| \sum_{j=1}^{t-1} \rho_n^{-j} \frac{  u_j^{(2)} }{ \sqrt{ \ell( \eta_n ) } } \right| >  \frac{ k_n x}{ 3 }  \right).
\end{align*}
Notice that we have that 
\begin{align*}
\mathbb{P} &\left( \underset{ 1 \leq t \leq n  }{ \mathsf{max} } \left| \sum_{j=1}^{t-1} \rho_n^{-j} u_j^{(2)} \right| > \frac{ k_n \sqrt{ \ell( \eta_n ) x}  }{3} \right)  
\\
&\leq \frac{ \displaystyle \mathbb{E} \left| \sum_{j=1}^{t-1} \rho_n^{-j} u_j^{(2)} \right|   }{ \frac{ k_n \sqrt{ \ell( \eta_n ) x}  }{3} }  
\leq o \left( \frac{ \ell ( \eta_n )  }{ \eta_n } \right) . \frac{ A }{ \left( \rho_n - 1 \right) k_n \sqrt{ \ell ( \eta_n )  } } 
=
o \left(  \frac{ \sqrt{ \ell ( \eta_n ) } }{ \eta_n } \right) \leq o \left( \frac{1}{ \sqrt{n} } \right) = o(1). 
\end{align*}

\newpage 

Therefore, we have that 
\begin{align}
\mathbb{P} \left(  \underset{ 1 \leq t \leq n  }{ \mathsf{max} } \left| \frac{ \rho_n^{-n} }{ k_n \sqrt{  \ell ( \eta_n ) } } x_{t-1} \right| > x \right) \leq \ \text{all the above} \ \to 0.    
\end{align}

\begin{proof}
To prove Lemma 1, using the Cramer-Wold device, it suffices to show that 
\begin{align}
a X_n + b Y_n \overset{ d }{ \to } \mathcal{N} \left( 0, \frac{ ( a^2 + b^2 ) \sigma^2  }{ 2c } \right), \ \text{for any} \ a, b \in \mathbb{R}.   
\end{align}
We rewrite $a X_n + b Y_n = \sum_{i=1}^n \zeta_{ni}$, where
\begin{align}
\zeta_{ni} = \frac{1}{ \sqrt{ n^{\gamma} } } \bigg[ a \rho_n^{-i} + b \rho_n^{ - (n-i) - 1 } \bigg] u_i, \ \ \ 1 \leq i \leq n,    
\end{align}
Denote with $r_n, q_n, q_n$ be sequences of positive integers such that 
\begin{align}
r_n \left( p_n + q_n \right) \leq n < ( r_n + 1 ) ( p_n + q_n )    
\end{align}
and $r_n \sim n^{ 1 - \nu / 2 }, p_n \sim n^{ \nu / 2 } - n^{ \nu / 4 }$ and $q_n \sim n^{ \nu / 4 }$. Notice that here notation $a_n \sim b_n$ means $\frac{ a_n }{ b_n } \to 1$ as $n \to \infty$. Therefore, we write
\begin{align}
\sum_{t=1}^n \zeta_{ni} = \frac{1}{ \sqrt{ n^{\gamma} } } \sum_{t=1}^n V_j +  \frac{1}{ \sqrt{ n^{\gamma} } }  \sum_{t=1}^n W_j  +  \frac{1}{ \sqrt{ n^{\gamma} } } R_n, 
\end{align}
Since $\left\{ V_1^{*},...., V_{ r_n }^{*} \right\}$ are independent, we have that as $n \to \infty$, 
\begin{align}
\mathbb{E} \left(  \sum_{j=1}^{ r_n } \frac{ V_j^{*} }{  \sqrt{ n^{\gamma} } }     \right)^2 \to \frac{ ( a^2 + b^2 ) \sigma^2  }{  2c }.     
\end{align}
Next, we have as $n \to \infty$
\begin{align}
\frac{ p_n }{ n^{\gamma} } \sum_{j=1}  \sum_{ \ell = - ( p_n - 1)}^{ ( p_n - 1) } \sum_{ k = 1}^{ p_n - | \ell | }  \big(  I_{n,j,\ell,k} + I_{n,\ell} \big) \to \frac{ a^2 + b^2 }{ 2 c }, 
\end{align}
which implies that the LHS is uniformly bounded by a positive real number $K_1$. 
\end{proof}

\begin{remark}
Notice that when the model parameter is such that $| \rho_n | < 1$, then the influence of the initial value of $\epsilon_0$ vanishes as $t$ grows. Therefore, it is convenient to take $\epsilon_0 = 0$.However, the error committed by imposing $\epsilon_0$ becomes serious when $| \rho_n | > 1$.   
\end{remark}

\newpage

\paragraph{Proof of Lemma 4.3:}

\

The following joint weakly functional convergence result holds: 
\begin{align}
\left( J_n(\tau), \frac{ \rho_n^{-2n} }{ k_n^2 \ell ( \eta_n ) } \sum_{t=1}^n y^2_{t-1}, \  \frac{ \rho_n^{-n} }{ \sqrt{ k^2_n  \ell ( \eta_n ) } } \sum_{t=1}^n y_{t-1} \psi_{\tau} \big( u_{t \tau} \big)   \right)  
\Rightarrow
\left(  J(\tau), \frac{1}{2c} L^2, K(\tau) L \right).
\end{align}
and it also holds that 
\begin{align}
\underset{ 1 \leq t \leq n }{ \mathsf{max} } \left|  \frac{ \rho_n^{-n} }{  \sqrt{ k_n^2 \ell ( \eta_n ) } } y_{t-1} \right| \overset{p}{\to}.   
\end{align}
Notice that the above vector-valued functional implies that the second coordinate is a functional of the population model parameter that corresponds to the data generating process being estimated without the presence of quantile-dependent parameters.  

\begin{proof}
\begin{align*}
\frac{ \rho_n^{-2n + 1}  }{ k_n \ell ( \eta_n ) } \sum_{t=1}^n y_{t-1} u_t &= 
\frac{ \rho_n^{-n + 1} y_0  }{  \sqrt{ k_n \ell ( \eta_n ) } } \frac{1}{ \sqrt{ k_n \ell (\eta_n) } } \sum_{t=1}^n \rho_n^{- (n-t) - 1} u_t + \frac{ \rho_n^{-2n + 1} }{ k_n \ell ( \eta_n ) } \sum_{t=1}^n \left( \sum_{j=1}^{t-1} \rho_n^{ t - 1 - j } u_j \right) u_t
\\
&= 
\frac{ \rho_n^{-2n + 1} }{ k_n \ell ( \eta_n ) } \sum_{t=1}^n  \left( \sum_{j=1}^{t-1} \rho_n^{ t - 1 - j } u_j \right) u_t + o_p(1). 
\end{align*}
Therefore, it suffices to show that the following probability bound result holds:
\begin{align}
\frac{ \rho_n^{-2n + 1} }{ k_n \ell ( \eta_n ) } \sum_{t=1}^n  \left( \sum_{j=1}^{t-1} \rho_n^{ t - 1 - j } u_j \right) u_t  = o_p(1).    
\end{align}
Here we have that 
\begin{align*}
\frac{ \rho_n^{-n}  }{ k_n  \sqrt{ \ell ( \eta_n ) }  } \sum_{t=1}^n y_{t-1} \psi_{\tau} \big( u_t \big) 
&=
\frac{ y_0 }{ k_n  \sqrt{ \ell ( \eta_n ) }  } \cdot \frac{1}{ \sqrt{k_n} }  \sum_{t=1}^n \rho_n^{ - ( n - t + 1 ) } \psi_{\tau} \big( u_t \big) + \frac{ \rho_n^{-n} }{ k_n  \sqrt{ \ell ( \eta_n ) }  } \sum_{t=1}^n \left( \sum_{j=1}^{t-1} \rho_n^{ t - 1 - j } u_j \right)  \psi_{\tau} \big( u_t \big)
\\
&=
\left( \frac{1}{ \sqrt{k_n} } \sum_{t=1}^n \rho_n^{ - ( n - t + 1 ) } \psi_{\tau} \big( u_t \big) \right) \cdot \left( \frac{1}{ \sqrt{ k_n \ell( \eta_n  ) } }  \sum_{t=1}^n \rho_n^{-j} u_j \right) +o_p(1)
\\
K_n( \tau ) L_n + o_p(1). 
\end{align*}
Therefore, it holds that 
\begin{align}
\mathbb{P} \left(  \underset{ 1 \leq t \leq n }{ \mathsf{max} } \left| \frac{ \rho_n^{-n} }{  k_n \sqrt{ \ell ( \eta_n ) } } y_{t-1} \right| > x \right) \to 0.     
\end{align}
\end{proof}

\newpage

\section{Main Results}

Recall that an $M-$estimator $\big( \hat{\mu}, \hat{\rho} \big)$ of the parameter vector $\big( \mu, \rho \big)$ is a minimizer for $\mu \in \mathbb{R}$ and $\rho \in \mathbb{R}$ such that 
\begin{align}
R_n ( \mu, \rho ) = \frac{1}{n} \sum_{t=1}^n \varrho_{\uptau} \left( y_t - \mu - \rho_c x_{t-1} \right). 
\end{align}
The formulation of the above expression covers the least squares regression when $\rho_{\uptau} \left( \mathsf{u} \right) = \mathsf{u} / 2$ and the quantile regression such that $\rho_{ \uptau } \left( \mathsf{u} \right) = \left\{ \uptau \mathsf{u}  \mathbf{1} \left( \mathsf{u} > 0 \right) - (1 - \uptau) \mathsf{u} \mathbf{1} \left( \mathsf{u} < 0 \right) \right\}$, for some $0 < \uptau < 1$.

\subsection{Near-Stationary Case $( c < 0 )$}

\medskip

\begin{lemma}
\label{lemma1}
Under Assumption \ref{Assumption1}, when $c < 0$ then it holds that 
\begin{align}
\mathbb{P} \left( \underset{ 1 \leq t \leq n }{ \mathsf{max} } \ y_t^2 \geq  \lambda \right) \leq \frac{ \mathbb{E} \left[ y_n^2 \right] }{ \lambda^2 },  \ \text{for some} \ \lambda > 0. 
\end{align}
\end{lemma}

\medskip

\begin{proof}
Notice that Lemma \ref{lemma1} corresponds to the Kolmogorov, Doob maximal inequality applied to the martingale sequence $\big( y_n, \mathcal{F}_n \big)_{ n \in \mathbb{N} }$. The proof of Lemma \ref{lemma1} can be easily obtained by considering the set $\mathcal{S}_s = \left\{ y_s^2 > \lambda , y_j \leq \lambda, j \leq s \right\}$ and expanding the expression for the expectation $\mathbb{E} \left[ y_n^2 \right]$. The particular result provides a probability bound for the tails of the autoregressive time series which inherits the properties of the stochastic difference equation.    
\end{proof}

In addition to Lemma \ref{lemma1}, the following two properties hold:
\begin{itemize}

\item If $c < 0$, $\mathbb{E} \big[ y_n^2 \big] = \mathcal{O} ( k_n )$.

\item $\underset{ 1 \leq t \leq n }{ \mathsf{max} } \ \displaystyle \frac{ y_t^2 }{ n } = o_p(1)$.

\end{itemize}

\medskip

\begin{lemma}
\label{lemma2}
Under Assumptions \ref{Assumption1} - \ref{Assumption3}, it holds that 
\begin{align}
\sum_{t=1}^n \mathbb{E}_{\mathcal{F}_{t-1}} \bigg[ \varphi_{nt} \big( \boldsymbol{\delta}( \uptau ) \big) \bigg] \overset{ p }{ \to } \frac{1 }{ 2 } \xi \times \boldsymbol{\delta}^{\top} ( \uptau ) \boldsymbol{B} \boldsymbol{\delta} ( \uptau ).
\end{align}
where 
\begin{align}
\xi  := \left| \frac{ \partial }{ \partial \theta } \mathbb{E} \bigg[ \psi \bigg( u_1 ( \uptau ) - \theta \bigg) \bigg] \right|_{ \theta = 0 } \ \ \ \text{and} 
\ \ \ 
\boldsymbol{B} = 
\begin{pmatrix}
1 & 0 
\\
0 & \sigma^2 / (-2c)
\end{pmatrix}.
\end{align}
\end{lemma}

\bigskip

\newpage 

\begin{proof} 
By rearranging expression \eqref{B.approx}, taking the conditional expectation and sum over $1 \leq t \leq n$, we obtain the following expression  
\begin{align}
\sum_{t=1}^n \mathbb{E}_{\mathcal{F}_{t-1}} \big[ \varphi_{nt} ( \boldsymbol{\delta} ) \big] 
=
\sum_{t=1}^n \mathbb{E}_{\mathcal{F}_{t-1}} \bigg[ \varrho_{\uptau} \left( \varepsilon_t - \boldsymbol{\delta}( \uptau )^{\top} \boldsymbol{X}_t \right) - \varrho_{\uptau} \left( \varepsilon_t \right) \bigg]. 
\end{align}

Using \cite{knight1998limiting}'s identity we have that 
\begin{align}
\varrho_{\uptau} ( \mathsf{u}_1 - \mathsf{u}_2 ) - \varrho_{\uptau} ( \mathsf{u}_1 ) 
= 
\mathsf{u}_2 \big( \uptau - \mathbf{1} \left\{ \mathsf{u}_1 \leq 0 \right\} \big) + \mathsf{u}_2 \int_0^1 \bigg[ \mathbf{1} \big\{ \mathsf{u}_1 \leq \mathsf{u}_2 s \big\}  - \mathbf{1} \big\{ \mathsf{u}_1 \leq 0 \big\} \bigg] ds.
\end{align}
which implies that we can decompose $Z_n( \mathsf{u}, \uptau ) = Z_n^{(1)} ( \mathsf{u}, \uptau ) + Z_n^{(2)} ( \mathsf{u}, \uptau )$. 
Then, it can be proved that 
\begin{align}
\sum_{t=1}^n \mathbb{E}_{\mathcal{F}_{t-1}} \bigg[ \varphi_{nt} ( \boldsymbol{\delta} ) \big) \bigg] 
=
\frac{1}{2} \xi  \times \boldsymbol{ \delta }^{\top} \left( \sum_{t=1}^n \boldsymbol{X}_t \boldsymbol{X}_t^{\top} \right) \boldsymbol{ \delta }  + o_p(1).
\end{align}
Furthermore, by Theorem 3.2 (a) of \cite{Phillips2007limit} it holds that 
\begin{align}
\frac{1}{n k_n } \sum_{t=1}^n  y_{t-1}^2 \overset{ p }{ \to } \frac{ \sigma^2 }{ - 2c }.
\end{align}
Therefore, it suffices to show that 
\begin{align}
\frac{1}{n \sqrt{k_n} } \sum_{t=1}^n y_{t-1} = o_p(1).
\end{align}
To see this, we consider the left side of the expression above such that
\begin{align}
\frac{ ( 1 - \rho_n ) }{ n \sqrt{k_n} ( 1 - \rho_n )  } \sum_{t=1}^n y_{t-1} 
= 
\frac{ 1 }{ n \sqrt{k_n} ( 1 - \rho_n )  } \sum_{t=1}^n ( 1 - \rho_n ) y_{t-1} 
= 
\frac{ 1 }{ n \sqrt{k_n} ( 1 - \rho_n )  } \sum_{t=1}^n y_{t-1} - \rho_n y_{t-1}
\end{align}
However, it holds that $y_t = \rho_n y_{t-1} + u_t$, and by rearranging we have that $-\rho_n y_{t-1} = - ( y_t - u_t )$. Thus,  
\begin{align*}
\frac{ ( 1 - \rho_n ) }{ n \sqrt{k_n} ( 1 - \rho_n )  } \sum_{t=1}^n y_{t-1} 
&=
\frac{ 1 }{ n \sqrt{k_n} ( 1 - \rho_n )  }
\sum_{t=1}^n \big[ y_{t-1} - \left( y_t - u_t \right) \big]
\\
&=
\frac{ 1 }{ n \sqrt{k_n} ( c / k_n ) } 
\sum_{t=1}^n \left( y_0 - y_n + \sum_{t=1}^n u_t \right) = o_p(1). 
\end{align*}
which shows that $\frac{ 1 }{ n \sqrt{k_n} } \sum_{t=1}^n y_{t-1} \overset{ p }{ \to } 0$, converges in probability to zero. 
\end{proof}

\medskip

\begin{lemma}
\label{lemma3}
Under Assumptions \ref{Assumption0}-\ref{Assumption3}, it holds that
\begin{align}
\sum_{t=1}^n \varphi_{nt} \big( \boldsymbol{\delta}( \uptau ) \big) \overset{ p }{ \to } \frac{1}{2} \xi  \times \boldsymbol{\delta} ( \uptau )^{\top} \boldsymbol{B} \boldsymbol{\delta} ( \uptau ).
\end{align}
\end{lemma}

\newpage 

\begin{proof}
By Lemma \ref{lemma2}, we can show that 
\begin{align}
\label{exprA}
\sum_{t=1}^n \bigg( \varphi_{nt} \big( \boldsymbol{\delta}( \uptau ) \big) - \mathbb{E}_{t-1} \bigg[ \varphi_{nt} \big( \boldsymbol{\delta}( \uptau ) \big)  \bigg]  \bigg) = o_p(1). 
\end{align}
Define the set $\mathcal{B}_t ( \lambda ) := \big\{ \frac{1}{n} y_{t-1}^2 \leq \lambda  \big\}$ for some positive $\lambda \in \mathbb{R}$. Then, \eqref{exprA} becomes as below
\begin{align}
\sum_{t=1}^n \bigg( \varphi_{nt} ( \boldsymbol{\delta} ) \mathbf{1} \big\{ \mathcal{B}_t ( \lambda )  \big\} - \mathbb{E}_{t-1} \bigg[ \varphi_{nt} ( \boldsymbol{\delta} ) \mathbf{1} \big\{ \mathcal{B}_t ( \lambda )  \big\}  \big) \bigg]  \bigg) = o_p(1). 
\end{align}

In particular, the expression $\bigg\{ \varphi_{nt} ( \boldsymbol{\delta} ) \mathbf{1} \big\{ \mathcal{B}_t ( \lambda )  \big\}  - \mathbb{E}_{t-1} \bigg[ \varphi_{nt} ( \boldsymbol{\delta} ) \mathbf{1} \big\{ \mathcal{B}_t ( \lambda )  \big\}  \bigg] \bigg\}_{t = 1}^n$ forms a martingale difference sequence, which by the Lenglart's inequality (see, \cite{jacod2003limit}) it follows that
\begin{align}
\sum_{t=1}^n \mathbb{E}_{t-1} \bigg[ \varphi_{nt}^2  ( \boldsymbol{\delta} ) \mathbf{1} \big\{ \mathcal{B}_t ( \lambda ) \big\} \bigg] = o_p(1). 
\end{align}  
Therefore, by definition of $\varphi_{nt}$ we have that 
\begin{align}
\varphi_{nt} ( \boldsymbol{\delta} ) = \boldsymbol{\delta}^{\top} \boldsymbol{X}_t \int_0^1 \bigg(  \psi ( \tilde{\varepsilon}_t ) - \psi \left( \tilde{\varepsilon}_t - \varepsilon_t \boldsymbol{\delta}^{\top} \boldsymbol{X}_t \right)  \bigg) d \varepsilon.  
\end{align}
Furthermore, by the non-decreasing property of $\psi(x)$ it holds that 
\begin{align}
\sum_{t=1}^n \mathbb{E}_{t-1} \bigg[ \varphi_{nt}^2  ( \boldsymbol{\delta} ) \mathbf{1} \big\{ \mathcal{B}_t ( \lambda ) \big\} \bigg] \leq \underset{ 1 \leq t \leq n }{ \mathsf{max} } \mathbb{E}_{t-1} \bigg[ \varphi_{nt}^2  ( \boldsymbol{\delta} ) \mathbf{1} \big\{ \mathcal{B}_t ( \lambda ) \big\} \bigg] \times \left( \sum_{t=1}^n \boldsymbol{\delta}^{\top} \boldsymbol{X}_t \boldsymbol{X}_t^{\top} \boldsymbol{\delta} \right) = o_p(1).
\end{align}
\end{proof}

\begin{corollary}
\begin{align}
\sum_{t=1}^n \bigg( \rho_{\uptau}  \left( \tilde{\varepsilon}_t - \boldsymbol{\delta}^{\top} \boldsymbol{X}_t \right) - \rho_{\uptau} \left( \varepsilon_t \right) \bigg) 
=
- \boldsymbol{\delta}^{\top} \sum_{t=1}^n \boldsymbol{X}_t \psi \left(  \varepsilon_t \right) + \frac{1}{2} \xi \boldsymbol{\delta}^{\top} \boldsymbol{B}   \boldsymbol{\delta} + R_n \left( \boldsymbol{\delta} \right)
\end{align}
with $R_n ( \boldsymbol{\delta} ) = o_p(1)$ for a fixed parameter vector $\boldsymbol{\delta}$ and $\underset{ \norm{  \boldsymbol{\delta} }  \leq C }{ \mathsf{max} } R_n( \boldsymbol{\delta} ) = o_p(1)$. 
\end{corollary}

\begin{proof}
Notice that the function $\varrho ( \mathsf{u} )$ is convex, therefore we can apply the same argument as that in the proof of Theorem 1 in \cite{pollard1991asymptotics} and show that an equivalent solution to the optimization problem is given by the following expression   
\begin{align}
\hat{\gamma} = \sum_{t=1}^n \frac{1}{\xi} \boldsymbol{B}^{-1} \psi ( \varepsilon_t ) \boldsymbol{X}_t^{\top} + o_p(1).
\end{align}
\end{proof}

\newpage

\subsection{Near-Explosive Case $( c > 0 )$}

\medskip

\begin{lemma}
\label{lemma4}
Consider that $y_1,...,y_n$ are random variables generated from the autoregressive process. Then, when $c > 0$ it holds that
\begin{align}
\mathbb{E} \big[ y_n^2 \big] = o \left( \rho_n^{2n} k_n^2   \right)
\end{align}
\end{lemma}
In addition to Lemma \ref{lemma4} the following two results hold
\begin{align}
\underset{ 1 \leq t \leq n }{ \mathsf{max} } \left\{ \frac{ y_t^2 }{ \rho_{n, c}^{2n} k_n^2 } \right\} 
&= o_p(1)
\\
\frac{ 1  }{ \sqrt{n} \rho_{n, c}^{n} k_n }  \sum_{t=1}^n y_{t-1} 
&= o_p(1).
\end{align}

\begin{lemma}
We consider the following two joint convergence results 
\begin{itemize}

\item[(\textit{i}).] Under Assumptions \ref{Assumption0}-\ref{Assumption3} it holds that 
\begin{align}
\left( \frac{1}{\sqrt{n}} \sum_{t=1}^n \psi_{\uptau}  \big( \varepsilon_t \big), \frac{1}{ \sqrt{k_n} } \sum_{t=1}^n \rho_n^{t - (n-1)} \psi_{\uptau}  \big( \varepsilon_t \big), \frac{1}{\sqrt{n}} \sum_{t=1}^n \rho_n^{-t} \varepsilon_t  \right)^{\top} \overset{ d }{ \to } \bigg( \mathcal{Z}_1, \mathcal{Z}_2, \mathcal{Z}_3  \bigg)^{\top}.
\end{align}
where $\bigg( \mathcal{Z}_1, \mathcal{Z}_2, \mathcal{Z}_3 \bigg)$ is a Gaussian random vector with independent components and the finite variance terms given by  $\mathbb{E} \left[ \psi_{\uptau}^2 ( \varepsilon_1 ) \right]$, $\frac{1}{2c} \mathbb{E} \left[ \psi_{\uptau}^2 ( \varepsilon_1 ) \right]$ and $\sigma^2 \big/ (2c)$, respectively. 

\item[(\textit{ii}).] Under Assumptions \ref{Assumption0}-\ref{Assumption3} it holds that 
 \begin{align}
\left( \frac{1}{\sqrt{n}} \sum_{t=1}^n \psi_{\uptau}  \big( \varepsilon_t \big), \frac{1}{ \rho_{n,c}^n k_n } \sum_{t=1}^n y_{t-1} \psi_{\uptau}  \big( \varepsilon_t \big), \frac{1}{ \rho_{n,c}^{2n} k_n^2 } \sum_{t=1}^n y_{t-1}^2 \right)^{\top} \overset{ d }{ \to } \bigg( \mathcal{Z}_1, \mathcal{Z}_2 \mathcal{Z}_3, \frac{ \mathcal{Z}_3^2 }{2c}  \bigg)^{\top}.
\end{align}
\end{itemize}
\end{lemma}

\begin{proof}
Recall that the difference equation with no model intercept, that is, $y_t = \rho_{n,c} y_{t-1} + \tilde{\varepsilon}_t$ has a general solution of the form $y_t = \rho^t_{n,c} y_0 + \sum_{j=1}^{t} \rho_{n,c}^{t-j} \tilde{\varepsilon}_{j}$. Similarly, for $y_{t-1}$ by shifting the time index such that  $t \mapsto t-1 $, then the equivalent general solution is given by the following expression 
\begin{align}
y_{t-1} = \rho^{t-1}_{n,c} y_0 + \sum_{j=1}^{t-1} \rho_{n,c}^{t-1-j} \tilde{\varepsilon}_{j}
\end{align}
Thus, by substituting the above expression to the sample moment $\sum_{t=1}^n y_{t-1} \psi_{\uptau}  \big( \varepsilon_t \big)$ we obtain that 
\begin{align}
\frac{1}{ \rho_{n,c}^n k_n } \sum_{t=1}^n y_{t-1} \psi_{\uptau}  \big( \varepsilon_t \big) 
&=
\frac{ y_0 }{ \rho_{n,c}^n k_n } \sum_{t=1}^n \rho_{n,c}^{t-1}  \psi_{\uptau} \big( \varepsilon_t \big) 
+ 
\frac{ 1 }{ \rho_{n,c}^n k_n } \sum_{t=1}^n \left( \sum_{j=1}^{t-1} \rho_{n,c}^{t-j-1} \tilde{\varepsilon}_{j} \right) \psi_{\uptau} \big( \varepsilon_t \big)
\end{align}
Then, since the first term of the above expression is asymptotically negligible by splitting the inner summation of the last term we obtain that 
\begin{align*}
\frac{1}{ \rho_{n,c}^n k_n } \sum_{t=1}^n y_{t-1} \psi_{\uptau}  \big( \varepsilon_t \big) 
&=
\frac{ 1 }{ \rho_{n,c}^n k_n } \sum_{t=1}^n \left( \sum_{j=1}^{n} \rho_{n,c}^{t-1-j} \tilde{\varepsilon}_j  - \sum_{j=t}^{n} \rho_{n,c}^{t-1-j} \tilde{\varepsilon}_j \right) \psi_{\uptau} \big( \varepsilon_t \big) + o_p(1).
\end{align*}
Since, $\sum_{j=t}^{n} \rho_{n,c}^{t-1-1} \tilde{\varepsilon}_j \overset{ p }{ \to } 0$, then it follows that 
\begin{align*}
\frac{1}{ \rho_{n,c}^n k_n } \sum_{t=1}^n y_{t-1} \psi_{\uptau}  \big( \varepsilon_t \big) 
&=
\frac{ 1 }{ \rho_{n,c}^n k_n } \sum_{t=1}^n \left( \sum_{j=1}^{n} \rho_{n,c}^{t-1-j} \tilde{\varepsilon}_j  \right) \psi_{\uptau} \big( \varepsilon_t \big) + o_p(1)
\\
&=
\left( \frac{1}{ \sqrt{ k_n } } \sum_{t=1}^n \rho_{n,c}^{ t-(n+1) } \psi_{\uptau} \big( \varepsilon_t \big) \right) \left(  \frac{1}{ \sqrt{ k_n } }  \sum_{t=1}^{n} \rho_{n,c}^{-t} \tilde{\varepsilon}_t \right) + o_p(1).
\end{align*}
Similarly, it holds that 
\begin{align}
\frac{1}{ \rho_{n,c}^{2n} k_n^2 } \sum_{t=1}^n y_{t-1}^2 
= 
\frac{1}{2c} \left(  \frac{1}{ \sqrt{ k_n } } \sum_{t=1}^{n} \rho_{n,c}^{-t} \tilde{\varepsilon}_t \right) + o_p(1).
\end{align}

\end{proof}

\begin{lemma}
The following joint convergence results hold 

\begin{itemize}

\item[\textit{(i)}.] Under Assumptions above we have that 
\begin{align*}
- \boldsymbol{\vartheta}_n ( \uptau )^{\top} \sum_{t=1}^n \boldsymbol{X}_t \psi_{\uptau} \big( \varepsilon_t \big)  + \frac{ \xi }{2} \boldsymbol{\vartheta}_n ( \uptau )^{\top} \left(  \sum_{t=1}^n \boldsymbol{X}_t \boldsymbol{X}_t^{\top} \right) \boldsymbol{\vartheta}_n ( \uptau )  \overset{ d }{ \to } - \boldsymbol{\vartheta}_n ( \uptau ) \big( \mathcal{Z}_1, \mathcal{Z}_2 \mathcal{Z}_3 \big)^{\top} + \frac{\xi}{2} \boldsymbol{\vartheta}_n ( \uptau )^{\top} \boldsymbol{B} \boldsymbol{\vartheta}_n ( \uptau )
\end{align*}

\item[\textit{(ii)}.] Under Assumptions above we have that 
\begin{align*}
\sum_{t=1}^n \left[ \varrho_{\uptau} \bigg( y_{t} - \boldsymbol{\vartheta}_n ( \uptau )^{\top} \boldsymbol{X}_{t} \bigg) - \varrho_{\uptau} \left( \varepsilon_t \right) \right] 
= 
- \boldsymbol{\vartheta}_n ( \uptau )^{\top} \sum_{t=1}^n \boldsymbol{X}_t \psi_{\uptau} \big( \varepsilon_t \big) + \frac{\xi}{2}  \boldsymbol{\vartheta}_n ( \uptau )^{\top} \left(  \sum_{t=1}^n \boldsymbol{X}_t \boldsymbol{X}_t^{\top} \right) \boldsymbol{\vartheta}_n ( \uptau ) + R_n \bigg( \boldsymbol{\vartheta}_n ( \uptau ) \bigg).
\end{align*} 
with  $R_n \bigg( \boldsymbol{\vartheta}_n ( \uptau ) \bigg) = o_p(1)$ for fixed $\boldsymbol{\vartheta}_n ( \uptau )$ and further  
\end{itemize}

\end{lemma}

\begin{proof}
In order to prove the uniformity condition we denote with 
\begin{align}
\varphi \big( \boldsymbol{\vartheta}_n (\uptau) \big) 
= 
\frac{1}{2} \xi \times \boldsymbol{\vartheta}_n (\uptau)^{\top} \left( \sum_{t=1}^n \boldsymbol{X}_t \boldsymbol{X}_t^{\top}  \right) \boldsymbol{\vartheta}_n (\uptau).
\end{align}
Furthermore, we need to show that 
\begin{align}
\underset{ \norm{ \boldsymbol{\vartheta}_n (\uptau) } \leq C }{ \mathsf{sup} } \left| \sum_{t=1}^n \varphi_{nt} \big( \boldsymbol{\vartheta}_n (\uptau) \big) - \varphi \big( \boldsymbol{\vartheta} (\uptau) \big)  \right| = o_P(1).
\end{align}
Since $\sum_{t=1}^n \boldsymbol{X}_t \boldsymbol{X}_t^{\top}$ converges in distribution, for any $\lambda > 0$ there exists $M$ large enough, such that, 
\begin{align}
\mathbb{P} \left(  \underset{ \norm{ \boldsymbol{\vartheta}_n (\uptau)  } \leq C }{ \mathsf{sup} }  \left|  \sum_{t=1}^n \varphi_{nt} \big( \boldsymbol{\vartheta}_n (\uptau) \big) - \varphi \big( \boldsymbol{\vartheta} (\uptau) \big) \right| \mathbf{1} \left\{ \norm{ \sum_{t=1}^n \boldsymbol{X}_t \boldsymbol{X}_t^{\top} } > M \right\} > \lambda / 2    \right) < \lambda^{*} /2.
\end{align}
On the other hand, on $\left\{ \norm{ \sum_{t=1}^n \boldsymbol{X}_t \boldsymbol{X}_t^{\top} } \right\}$, for any $\lambda  > 0$, there exists $\delta > 0$, such that, 
\begin{align}
\underset{ \norm{ \boldsymbol{ \vartheta } } \leq C }{ \mathsf{sup} } \bigg| \varphi ( \gamma + \vartheta ) - \varphi ( \gamma ) \bigg| \leq \lambda. 
\end{align}

Moreover, following the convexity Lemma of \cite{pollard1991asymptotics}, one can show that 
\begin{align}
\mathbb{P} \left( \underset{ \norm{ \boldsymbol{\vartheta}_n (\uptau) } \leq C }{ \mathsf{sup} }  \left|  \sum_{t=1}^n \varphi_{nt} \big( \boldsymbol{\vartheta}_n (\uptau) \big) - \varphi \big( \boldsymbol{\vartheta} (\uptau) \big) \right| \mathbf{1} \left\{ \norm{ \sum_{t=1}^n \boldsymbol{X}_t \boldsymbol{X}_t^{\top} } \leq M \right\} > \lambda / 2    \right) < \lambda^{*} /2.
\end{align}
Therefore, the combination of the above yields the uniformity result of interest. 
\end{proof}

\subsection{Weak Convergence of Functionals Results}

Following \cite{Koenker2005}, we consider that all parameters share the same monotone behaviour with respect to the quantile level $\uptau \in (0,1)$. Then the consistency of quantile dependent parameter can be deduced from the monotonicity of the subgradient as well as a direct consequence of the uniform convergence of the empirical distribution function and the Glivenko-Cantelli Theorem. Therefore, the asymptotic behaviour of $\sqrt{n} \left( \hat{\boldsymbol{\beta}}( \uptau ) - \boldsymbol{\beta} ( \uptau  ) \right)$ follows by considering the following objective function
\begin{align}
\mathcal{Z}_n ( \delta ) 
= 
\frac{1}{n} \sum_{t=1}^n \varrho_{\uptau} \left( \varepsilon_t - \boldsymbol{X}_t^{\top} \boldsymbol{\delta} / \sqrt{n} \right) - \varrho_{\uptau} \big(  \varepsilon_t \big)
\end{align} 
where $\varepsilon_t ( \uptau ) = y_t - \boldsymbol{X}_t^{\top} \boldsymbol{\beta}( \uptau )$. The function $\mathcal{Z}_n ( \boldsymbol{\delta} )$ is convex, and is minimized at $\hat{\boldsymbol{\delta}}_n = \sqrt{n} \left( \hat{ \boldsymbol{\beta} } (\uptau) - \boldsymbol{\beta} ( \uptau )  \right)$. Therefore, we can show that the limit distribution of $\hat{ \boldsymbol{\delta} }_n$ can be determined by the asymptotic distribution of the objective function $\mathcal{Z}_n ( \boldsymbol{\delta} )$. Furthermore, it follows from the Lindeberg-Feller central limit theorem that $\mathcal{Z}_n^{(1)} \overset{ d }{ \to } -  \boldsymbol{\delta}^{\prime} \mathcal{W}$, where $\mathcal{W} \overset{ d }{ \to } \mathcal{N} \big( 0, \uptau (1 - \uptau) \boldsymbol{D}_0 \big)$. Then, following the proof of Theorem 4.1 of \cite{Koenker2005} it holds that (see also derivations in \cite{knight1998limiting} and \cite{kato2009asymptotics})  
\begin{align}
\sum_{t=1}^n \mathbb{E} \left[ \mathcal{Z}_{nt}^{(2)} ( \boldsymbol{\delta} ) \right] \overset{ d }{ \to } \frac{1}{2} \boldsymbol{\delta}^{\top} \boldsymbol{D}_1 \boldsymbol{\delta}.
\end{align}
Therefore, it can be proved that 
\begin{align}
\mathcal{Z}_n ( \boldsymbol{\delta} ) \overset{ d }{ \to } \mathcal{Z}_0 ( \boldsymbol{\delta} ) \equiv - \boldsymbol{\delta}^{\top} \mathcal{W} + \frac{1}{2} \boldsymbol{\delta}^{\top} \boldsymbol{D}_1 \boldsymbol{\delta}.
\end{align}
The convexity of the limiting distribution of $\mathcal{Z}_0 ( \boldsymbol{\delta} )$  ensures that the uniqueness of the minimizer 
\begin{align}
\hat{\delta}_n :=  \underset{ \delta }{ \mathsf{arg \min} }  \ \mathcal{Z}_n ( \boldsymbol{\delta} ) \mapsto \hat{\delta}_0 := \underset{ \boldsymbol{\delta} }{ \mathsf{arg \min} }  \ \mathcal{Z}_0 ( \boldsymbol{\delta} )
\end{align} 
where $\mathcal{Z}_n ( \boldsymbol{\delta} ) = \mathcal{Z}_n^{(1)} ( \boldsymbol{\delta} ) + \mathcal{Z}_n^{(2)} ( \boldsymbol{\delta} )$ and $\mathcal{Z}_0 ( \boldsymbol{\delta} ) = - \boldsymbol{\delta}^{\top} \mathcal{W} + \frac{1}{2} \boldsymbol{\delta}^{\top} \boldsymbol{D}_1 \boldsymbol{\delta}$ such that $\hat{ \boldsymbol{\delta} }_n = \sqrt{n} \left( \hat{ \boldsymbol{\beta} } (\uptau) - \boldsymbol{\beta} ( \uptau ) \right)$. A similar approach is followed in the study of \cite{mao2019moderate} who prove that for any $\uptau \in (0,1)$ the solution of the following expression is obtained by $
\frac{ \sqrt{n} }{ \lambda(n) } \left( \widehat{ \boldsymbol{\beta} }( \uptau ) - \boldsymbol{\beta} ( \uptau ) \right) \in \underset{ \mathsf{u} \in \mathbb{R}^p }{ \mathsf{arg \ min} } \ \mathcal{Z}_n ( \mathsf{u}, \uptau )$.

Then, it follows that  
\begin{align*}
\bigg| \mathcal{Z}_n( \mathsf{u}, \uptau ) - G_n( \mathsf{u}, \uptau ) \bigg| 
&= 
\left| \mathcal{Z}_n^{(2)} ( \mathsf{u} ,\uptau) - \frac{ \mathsf{u}^{\top} \boldsymbol{D} \mathsf{u} }{2} f \left( F^{-1} (\uptau) \right) \right|
\\
&\leq 
\left| \mathbb{E} \left[ \mathcal{Z}_n^{(2)} ( \mathsf{u},\uptau) \right] - \frac{ \mathsf{u}^{\top} \boldsymbol{D} \mathsf{u} }{2} f_{\varepsilon} \left( F_{\varepsilon}^{-1} (\uptau) \right) \right| 
+ \bigg| \mathcal{Z}_n^{(2)} (\mathsf{u},\uptau) - \mathbb{E} \left[ \mathcal{Z}_n^{(2)} (\mathsf{u},\uptau) \right] \bigg|
\end{align*}
where 
\begin{align}
\mathcal{Z}_n^{(2)} (\mathsf{u},\uptau) = \frac{1}{ \sqrt{n} \lambda (n) } \sum_{i=1}^n \int_0^1 \left( \mathbf{1} \left\{ \varepsilon_i \leq F_{\varepsilon}^{-1} (\uptau) + \frac{ \lambda (n) \sigma_{in}^{-1} x_{in}^{\top} \mathsf{u} s
 }{ \sqrt{n} } \right\} - \mathbf{1} \left\{ \varepsilon_i \leq F_{\varepsilon}^{-1} (\uptau) \right\} \right) ds.
\end{align}
Therefore, it can be shown that 
\begin{align}
\underset{ \uptau \in [ \alpha, 1 - \alpha ] }{ \mathsf{sup} } \left| \mathbb{E} \left[ \mathcal{Z}_n^{(2)} (\mathsf{u},\uptau) \right]  - \frac{ \mathsf{u}^{\top} \boldsymbol{D} \mathsf{u} }{2} f_{\varepsilon} \left( F_{\varepsilon}^{-1} (\uptau) \right) \right| \to 0, \ \ \text{as} \ n \to \infty
\end{align}

Thus, for large $n$, we further have that 
\begin{align}
\mathbb{P} \left( \underset{ \uptau \in [ \alpha, 1 - \alpha ] }{ \mathsf{sup} } \ \bigg| \mathcal{Z}_n( \mathsf{u},\uptau) - G_n(\mathsf{u},\uptau) \bigg|  \geq \lambda  \right) \leq \mathbb{P} \left( \underset{ \uptau \in [ \alpha, 1 - \alpha ] }{ \mathsf{sup} } \bigg| \mathcal{Z}^{(2)}_n(\mathsf{u},\uptau) - \mathbb{E} \left[ \mathcal{Z}^{(2)}_n(\mathsf{u},\uptau) \right] \bigg| \geq \frac{\lambda}{2}  \right) 
\end{align}
The following step is to prove that 
\begin{align}
\underset{ n \to \infty }{ \mathsf{ lim \ sup} } \ \frac{1}{ \lambda^2 (n) } \mathsf{log} \left\{ \mathbb{P} \left( \underset{ \uptau \in [ \alpha, 1 - \alpha ] }{ \mathsf{sup} } \bigg| \mathcal{Z}^{(2)}_n(\mathsf{u},\uptau) - \mathbb{E} \left[ \mathcal{Z}^{(2)}_n(\mathsf{u},\uptau) \right] \bigg| \geq \lambda \right) \right\} = - \infty
\end{align}
From the definition of $\mathcal{Z}^{(2)}_n(\mathsf{u},t)$, we obtain that 
\begin{align}
\mathcal{Z}^{(2)}_n( \mathsf{u},\uptau ) - \mathbb{E} \left[ \mathcal{Z}^{(2)}_n( \mathsf{u},\uptau) \right] = \int_0^1 \bigg( \mathcal{W}_n( \mathsf{u} s, \uptau ) - \mathcal{W}_n ( 0 , \uptau ) \bigg) ds,
\end{align}
where 
\begin{align*}
\mathcal{W}_n ( r, \uptau ) = \frac{ \sum_{t=1}^n \boldsymbol{X}_{nt}^{\prime} \mathsf{u} }{ \sqrt{n} \lambda (n) } \left[ \mathbf{1} \left\{ \varepsilon_i \leq F_{\varepsilon}^{-1} ( \uptau ) + \frac{ \lambda (n) \sigma_{n}^{-1} \boldsymbol{X}_{nt}^{\prime} r }{ \sqrt{n} } \right\} - F_{\varepsilon} \left(  F_{\varepsilon}^{-1} ( \uptau ) + \frac{ \lambda (n) \sigma_{n}^{-1}  \boldsymbol{X}_{nt}^{\prime} r }{ \sqrt{n} } \right) \right]
\end{align*}

\medskip

In practise we employ Theorem 1 from \cite{knight1998limiting}. 
Therefore, we have that 
\begin{align}
\mathcal{Z}_n^{(1)} ( \mathsf{u} ) = - \frac{1}{ \sqrt{n} } \sum_{i=1}^n \boldsymbol{x}_i \mathsf{u} \bigg[ \mathbf{1} \big( \epsilon_i > 0 \big) -  \mathbf{1} \big( \epsilon_i < 0 \big) \bigg]
\end{align}
and 
\begin{align}
\mathcal{Z}_n^{(2)} ( \mathsf{u} ) 
&= 
\frac{2 a_n }{ \sqrt{n} } \sum_{t=1}^n  \int_0^{ v_{n_i} }  \boldsymbol{X}_t \mathsf{u} \bigg[ \mathbf{1} \big( \epsilon_i < s \big) -  \mathbf{1} \big( \epsilon_i < 0 \big) \bigg] ds
\end{align} 
where $v_{n} = \boldsymbol{X}_t^{\top} \mathsf{u} / a_n$.

Then, by the Lindeberg-Feller central limit theorem, for each $\mathsf{u}$ it holds that, 
\begin{align}
\mathcal{Z}_n^{(1)} ( \mathsf{u} )  \overset{ d }{ \to } - \mathsf{u} \boldsymbol{W}
\end{align}
and the convergence in distribution holds for any finite collection of $\mathsf{u}'$s. For $\mathcal{Z}_n^{(2)} ( \mathsf{u} )$, we have that 
\begin{align}
\mathcal{Z}_n^{(2)} ( \mathsf{u} ) 
= 
\sum_{t=1}^n \mathbb{E} \left[ \mathcal{Z}_{nt}^{(2)} ( \mathsf{u} ) \right] + \sum_{t=1}^n \bigg[ \mathcal{Z}_{nt}^{(2)} ( \mathsf{u} ) - \mathbb{E} \left[ \mathcal{Z}_{nt}^{(2)} ( \mathsf{u} ) \right] \bigg]. 
\end{align}

We employ the above orthogonal decomposition when proving the limit results for the estimator of the quantile autoregressive model for moderate deviations from the unit boundary. Notice that the implementation of the corresponding FCLT for the quantile-dependent innovation term is applicable for the  \textit{i.i.d} innovation sequence assumption. An extension of the particular results to the case in which innovations have serial correlation via the use of a linear process representation for instance is also possible.  Further applications with suitable econometric conditions we could consider within our framework are presented in the study of   \cite{doukhan1999new} who consider a different type of weak dependence condition for time series models.

\newpage 
   
\bibliographystyle{apalike}
\bibliography{myreferences1}

\begin{thebibliography}{}

\bibitem[Abadir and Lucas, 2000]{abadir2000quantiles}
Abadir, K.~M. and Lucas, A. (2000).
\newblock Quantiles for t-statistics based on m-estimators of unit roots.
\newblock {\em Economics Letters}, 67(2):131--137.

\bibitem[Anderson, 1959]{anderson1959asymptotic}
Anderson, T.~W. (1959).
\newblock On asymptotic distributions of estimates of parameters of stochastic
  difference equations.
\newblock {\em The Annals of Mathematical Statistics}, pages 676--687.

\bibitem[Arvanitis and Magdalinos, 2018]{arvanitis2018mildly}
Arvanitis, S. and Magdalinos, T. (2018).
\newblock Mildly explosive autoregression under stationary conditional
  heteroskedasticity.
\newblock {\em Journal of Time Series Analysis}, 39(6):892--908.

\bibitem[Aue and Horv{\'a}th, 2007]{aue2007limit}
Aue, A. and Horv{\'a}th, L. (2007).
\newblock A limit theorem for mildly explosive autoregression with stable
  errors.
\newblock {\em Econometric Theory}, 23(2):201--220.

\bibitem[Bahadur, 1966]{bahadur1966note}
Bahadur, R.~R. (1966).
\newblock A note on quantiles in large samples.
\newblock {\em The Annals of Mathematical Statistics}, 37(3):577--580.

\bibitem[Benke and Pap, 2021]{benke2021nearly}
Benke, J.~M. and Pap, G. (2021).
\newblock Nearly unstable family of stochastic processes given by stochastic
  differential equations with time delay.
\newblock {\em Journal of Statistical Planning and Inference}, 211:1--11.

\bibitem[Billingsley, 1968]{billingsley1968convergence}
Billingsley, P. (1968).
\newblock {\em Convergence of probability measures}.
\newblock John Wiley \& Sons.

\bibitem[Buchmann and Chan, 2007]{Buchmann2007asymptotic}
Buchmann, B. and Chan, N.~H. (2007).
\newblock Asymptotic theory of least squares estimators for nearly unstable
  processes under strong dependence.
\newblock {\em The Annals of statistics}, 35(5):2001--2017.

\bibitem[Cai et~al., 2022]{cai2022new}
Cai, Z., Chen, H., and Liao, X. (2022).
\newblock A new robust inference for predictive quantile regression.
\newblock {\em Journal of Econometrics}.

\bibitem[Cavaliere, 2002]{cavaliere2002bounded}
Cavaliere, G. (2002).
\newblock Bounded integrated processes and unit root tests.
\newblock {\em Statistical Methods and Applications}, 11(1):41--69.

\bibitem[Cavanagh, 1985]{Cavanagh1985LUR}
Cavanagh, C. (1985).
\newblock Roots local to unity.
\newblock {\em Manuscript}.

\bibitem[Chan, 1988]{chan1988parameter}
Chan, N.~H. (1988).
\newblock The parameter inference for nearly nonstationary time series.
\newblock {\em Journal of the American Statistical Association},
  83(403):857--862.

\bibitem[Chan, 1990]{chan1990inference}
Chan, N.~H. (1990).
\newblock Inference for near-integrated time series with infinite variance.
\newblock {\em Journal of the American Statistical Association},
  85(412):1069--1074.

\bibitem[Chan et~al., 2006]{chan2006quantile}
Chan, N.~H., Peng, L., and Qi, Y. (2006).
\newblock Quantile inference for near-integrated autoregressive time series
  with infinite variance.
\newblock {\em Statistica Sinica}, pages 15--28.

\bibitem[Chan and Tran, 1989]{chan1989first}
Chan, N.~H. and Tran, L.~T. (1989).
\newblock On the first-order autoregressive process with infinite variance.
\newblock {\em Econometric Theory}, 5(3):354--362.

\bibitem[Chan and Wei, 1987]{chan1987asymptotic1}
Chan, N.~H. and Wei, C.-Z. (1987).
\newblock Asymptotic inference for nearly nonstationary ar (1) processes.
\newblock {\em The Annals of Statistics}, pages 1050--1063.

\bibitem[Cox and Llatas, 1991]{cox1991maximum}
Cox, D.~D. and Llatas, I. (1991).
\newblock Maximum likelihood type estimation for nearly nonstationary
  autoregressive time series.
\newblock {\em The Annals of Statistics}, pages 1109--1128.

\bibitem[Cram{\'e}r, 1951]{cramer1951contribution}
Cram{\'e}r, H. (1951).
\newblock A contribution to the theory of stochastic processes.
\newblock In {\em Proceedings of the Second Berkeley Symposium on Mathematical
  Statistics and Probability}, pages 329--339. University of California Press.

\bibitem[Cs{\"o}rg{\H{o}}, 1983]{csorgHo1983quantile}
Cs{\"o}rg{\H{o}}, M. (1983).
\newblock {\em Quantile processes with statistical applications}.
\newblock SIAM.

\bibitem[Csorgo et~al., 1986]{csorgo1986weighted}
Csorgo, M., Csorgo, S., Horv{\'a}th, L., and Mason, D.~M. (1986).
\newblock Weighted empirical and quantile processes.
\newblock {\em The Annals of Probability}, pages 31--85.

\bibitem[Dickey and Fuller, 1979]{dickey1979distribution}
Dickey, D.~A. and Fuller, W.~A. (1979).
\newblock Distribution of the estimators for autoregressive time series with a
  unit root.
\newblock {\em Journal of the American statistical association},
  74(366a):427--431.

\bibitem[Dickey and Fuller, 1981]{dickey1981likelihood}
Dickey, D.~A. and Fuller, W.~A. (1981).
\newblock Likelihood ratio statistics for autoregressive time series with a
  unit root.
\newblock {\em Econometrica: journal of the Econometric Society}, pages
  1057--1072.

\bibitem[Doukhan and Louhichi, 1999]{doukhan1999new}
Doukhan, P. and Louhichi, S. (1999).
\newblock A new weak dependence condition and applications to moment
  inequalities.
\newblock {\em Stochastic processes and their applications}, 84(2):313--342.

\bibitem[Duffy and Kasparis, 2021]{duffy2021estimation}
Duffy, J.~A. and Kasparis, I. (2021).
\newblock Estimation and inference in the presence of fractional d= 1/2 and
  weakly nonstationary processes.
\newblock {\em The Annals of Statistics}, 49(2):1195--1217.

\bibitem[Fan and Lee, 2019]{fan2019predictive}
Fan, R. and Lee, J.~H. (2019).
\newblock Predictive quantile regressions under persistence and conditional
  heteroskedasticity.
\newblock {\em Journal of Econometrics}, 213(1):261--280.

\bibitem[Fountis and Dickey, 1989]{fountis1989testing}
Fountis, N.~G. and Dickey, D.~A. (1989).
\newblock Testing for a unit root nonstationarity in multivariate
  autoregressive time series.
\newblock {\em The Annals of Statistics}, pages 419--428.

\bibitem[Fu et~al., 2022]{fu2022cqr}
Fu, K.-A., Ni, J., and Dong, Y. (2022).
\newblock Cqr-based inference for the infinite-variance nearly nonstationary
  autoregressive models.
\newblock {\em Lithuanian Mathematical Journal}, pages 1--9.

\bibitem[Garderen, 1999]{garderen1999exact}
Garderen, K. J.~v. (1999).
\newblock Exact geometry of autoregressive models.
\newblock {\em Journal of time series analysis}, 20(1):1--21.

\bibitem[Giraitis and Phillips, 2006]{giraitis2006uniform}
Giraitis, L. and Phillips, P.~C. (2006).
\newblock Uniform limit theory for stationary autoregression.
\newblock {\em Journal of time series analysis}, 27(1):51--60.

\bibitem[Goh and Knight, 2009]{goh2009nonstandard}
Goh, S.~C. and Knight, K. (2009).
\newblock Nonstandard quantile-regression inference.
\newblock {\em Econometric Theory}, 25(5):1415--1432.

\bibitem[H{\"a}rdle et~al., 2016]{hardle2016tenet}
H{\"a}rdle, W.~K., Wang, W., and Yu, L. (2016).
\newblock Tenet: Tail-event driven network risk.
\newblock {\em Journal of Econometrics}, 192(2):499--513.

\bibitem[Hasan and Koenker, 1997]{hasan1997robust}
Hasan, M.~N. and Koenker, R.~W. (1997).
\newblock Robust rank tests of the unit root hypothesis.
\newblock {\em Econometrica: Journal of the Econometric Society}, pages
  133--161.

\bibitem[Hirukawa and Lee, 2021]{hirukawa2021asymptotic}
Hirukawa, J. and Lee, S. (2021).
\newblock Asymptotic properties of mildly explosive processes with locally
  stationary disturbance.
\newblock {\em Metrika}, 84(4):511--534.

\bibitem[Huang et~al., 2014]{huang2014limit}
Huang, S.-H., Pang, T.-X., and Weng, C. (2014).
\newblock Limit theory for moderate deviations from a unit root under
  innovations with a possibly infinite variance.
\newblock {\em Methodology and Computing in Applied Probability},
  16(1):187--206.

\bibitem[Hui et~al., 2022]{Jian2022}
Hui, J., Yilong, W., and Guangyu, Y. (2022).
\newblock Deviation inequalities and cramer-type moderate deviations for the
  explosive autoregressive process.
\newblock {\em Bernoulli}, 28(5):1--28.

\bibitem[Hwang and Pang, 2009]{hwang2009asymptotic}
Hwang, K.-S. and Pang, T.-X. (2009).
\newblock Asymptotic inference for nearly nonstationary ar (1) processes with
  possibly infinite variance.
\newblock {\em Statistics \& probability letters}, 79(22):2374--2379.

\bibitem[Jacod and Shiryaev, 2003]{jacod2003limit}
Jacod, J. and Shiryaev, A. (2003).
\newblock {\em Limit theorems for stochastic processes}, volume 288.
\newblock Berlin: Springer - Verlag.

\bibitem[Jansson, 2004]{jansson2004error}
Jansson, M. (2004).
\newblock The error in rejection probability of simple autocorrelation robust
  tests.
\newblock {\em Econometrica}, 72(3):937--946.

\bibitem[Jansson and Moreira, 2006]{jansson2006optimal}
Jansson, M. and Moreira, M.~J. (2006).
\newblock Optimal inference in regression models with nearly integrated
  regressors.
\newblock {\em Econometrica}, 74(3):681--714.

\bibitem[Jiang et~al., 2015]{jiang2015moderate}
Jiang, H., Yu, M., and Yang, G. (2015).
\newblock Moderate deviations for the mildly stationary autoregressive models
  with dependent errors.
\newblock {\em arXiv preprint arXiv:1510.02862}.

\bibitem[Jure{\v{c}}kov{\'a} et~al., 1988]{jurevckova1988moderate}
Jure{\v{c}}kov{\'a}, J., Kallenberg, W., and Veraverbeke, N. (1988).
\newblock Moderate and cram{\'e}r-type large deviation theorems for
  m-estimators.
\newblock {\em Statistics \& probability letters}, 6(3):191--199.

\bibitem[Kato, 2009]{kato2009asymptotics}
Kato, K. (2009).
\newblock Asymptotics for argmin processes: Convexity arguments.
\newblock {\em Journal of Multivariate Analysis}, 100(8):1816--1829.

\bibitem[Katsouris, 2023]{katsouris2023structural}
Katsouris, C. (2023).
\newblock Structural break detection in quantile predictive regression models
  with persistent covariates.
\newblock {\em arXiv preprint arXiv:2302.05193}.

\bibitem[Kiefer et~al., 2000]{kiefer2000simple}
Kiefer, N.~M., Vogelsang, T.~J., and Bunzel, H. (2000).
\newblock Simple robust testing of regression hypotheses.
\newblock {\em Econometrica}, 68(3):695--714.

\bibitem[Knight, 1987]{knight1987rate}
Knight, K. (1987).
\newblock Rate of convergence of centred estimates of autoregressive parameters
  for infinite variance autoregressions.
\newblock {\em Journal of time series analysis}, 8(1):51--60.

\bibitem[Knight, 1998]{knight1998limiting}
Knight, K. (1998).
\newblock Limiting distributions for l1 regression estimators under general
  conditions.
\newblock {\em Annals of statistics}, pages 755--770.

\bibitem[Koenker, 2005]{Koenker2005}
Koenker, R. (2005).
\newblock Quantile regression.

\bibitem[Koenker and Bassett, 1978]{koenker1978regression}
Koenker, R. and Bassett, G. (1978).
\newblock Regression quantiles.
\newblock {\em Econometrica: journal of the Econometric Society}, pages 33--50.

\bibitem[Koenker and Portnoy, 1987]{koenker1987estimation}
Koenker, R. and Portnoy, S. (1987).
\newblock L-estimation for linear models.
\newblock {\em Journal of the American statistical Association},
  82(399):851--857.

\bibitem[Koenker and Xiao, 2002]{koenker2002inference}
Koenker, R. and Xiao, Z. (2002).
\newblock Inference on the quantile regression process.
\newblock {\em Econometrica}, 70(4):1583--1612.

\bibitem[Koenker and Xiao, 2004]{koenker2004unit}
Koenker, R. and Xiao, Z. (2004).
\newblock Unit root quantile autoregression inference.
\newblock {\em Journal of the American Statistical Association},
  99(467):775--787.

\bibitem[Koenker and Xiao, 2006]{koenker2006quantile}
Koenker, R. and Xiao, Z. (2006).
\newblock Quantile autoregression.
\newblock {\em Journal of the American statistical association},
  101(475):980--990.

\bibitem[Kong, 2015]{kong2015m}
Kong, X.-B. (2015).
\newblock M-estimation for moderate deviations from a unit root.
\newblock {\em Communications in Statistics-Theory and Methods},
  44(3):476--485.

\bibitem[Kostakis et~al., 2015]{kostakis2015Robust}
Kostakis, A., Magdalinos, T., and Stamatogiannis, M.~P. (2015).
\newblock Robust econometric inference for stock return predictability.
\newblock {\em The Review of Financial Studies}, 28(5):1506--1553.

\bibitem[Kyprianou, 2014]{kyprianou2014fluctuations}
Kyprianou, A.~E. (2014).
\newblock {\em Fluctuations of L{\'e}vy processes with applications:
  Introductory Lectures}.
\newblock Springer Science \& Business Media.

\bibitem[Lai and Wei, 1982]{lai1982least}
Lai, T.~L. and Wei, C.~Z. (1982).
\newblock Least squares estimates in stochastic regression models with
  applications to identification and control of dynamic systems.
\newblock {\em The Annals of Statistics}, 10(1):154--166.

\bibitem[Larsson, 1995]{larsson1995asymptotic}
Larsson, R. (1995).
\newblock The asymptotic distributions of some test statistics in
  near-integrated ar processes.
\newblock {\em Econometric Theory}, 11(2):306--330.

\bibitem[Lee, 2016]{lee2016predictive}
Lee, J.~H. (2016).
\newblock Predictive quantile regression with persistent covariates: Ivx-qr
  approach.
\newblock {\em Journal of Econometrics}, 192(1):105--118.

\bibitem[Lee, 2018]{lee2018limit}
Lee, J.~H. (2018).
\newblock Limit theory for explosive autoregression under conditional
  heteroskedasticity.
\newblock {\em Journal of Statistical Planning and Inference}, 196:30--55.

\bibitem[Ling and McAleer, 2004]{ling2004regression}
Ling, S. and McAleer, M. (2004).
\newblock Regression quantiles for unstable autoregressive models.
\newblock {\em Journal of Multivariate Analysis}, 89(2):304--328.

\bibitem[Liu and Liu, 2018]{liu2018limit}
Liu, Q. and Liu, X. (2018).
\newblock Limit theory for an ar (1) model with intercept and a possible
  infinite variance.
\newblock {\em arXiv preprint arXiv:1802.10299}.

\bibitem[Liu et~al., 2021]{liu2021quantile}
Liu, Q.-m., Liao, G.-l., and Zhang, R.-m. (2021).
\newblock Quantile inference for nonstationary processes with infinite variance
  innovations.
\newblock {\em Applied Mathematics-A Journal of Chinese Universities},
  36(3):443--461.

\bibitem[Liu et~al., 2022]{liu2022mildly}
Liu, X., Li, X., Gao, M., and Yang, W. (2022).
\newblock Mildly explosive autoregression with strong mixing errors.
\newblock {\em Entropy}, 24(12):1730.

\bibitem[Liu et~al., 2023]{liu2023unified}
Liu, X., Long, W., Peng, L., and Yang, B. (2023).
\newblock A unified inference for predictive quantile regression.
\newblock {\em Journal of the American Statistical Association}, pages 1--15.

\bibitem[Lucas, 1995]{lucas1995unit}
Lucas, A. (1995).
\newblock Unit root tests based on m estimators.
\newblock {\em Econometric Theory}, 11(2):331--346.

\bibitem[Lui et~al., 2018]{lui2018mild}
Lui, Y.~L., Xiao, W., and Yu, J. (2018).
\newblock Mild-explosive and local-to-mild-explosive autoregressions with
  serially correlated errors.

\bibitem[Lui et~al., 2021]{lui2021mildly}
Lui, Y.~L., Xiao, W., and Yu, J. (2021).
\newblock Mildly explosive autoregression with anti-persistent errors.
\newblock {\em Oxford Bulletin of Economics and Statistics}, 83(2):518--539.

\bibitem[Magdalinos, 2012]{magdalinos2012mildly}
Magdalinos, T. (2012).
\newblock Mildly explosive autoregression under weak and strong dependence.
\newblock {\em Journal of Econometrics}, 169(2):179--187.

\bibitem[Magdalinos and Petrova, 2022]{Magdalinos2022uniform}
Magdalinos, T. and Petrova, K. (2022).
\newblock Uniform and distribution-free inference with general autoregressive
  processes.

\bibitem[Magdalinos and Phillips, 2009]{Magdalinos2009limit}
Magdalinos, T. and Phillips, P. C.~B. (2009).
\newblock Limit theory for cointegrated systems with moderately integrated and
  moderately explosive regressors.
\newblock {\em Econometric Theory}, 25(2):482--526.

\bibitem[Mann and Wald, 1943]{mann1943statistical}
Mann, H.~B. and Wald, A. (1943).
\newblock On the statistical treatment of linear stochastic difference
  equations.
\newblock {\em Econometrica, Journal of the Econometric Society}, pages
  173--220.

\bibitem[Mao and Guo, 2019]{mao2019moderate}
Mao, M. and Guo, W. (2019).
\newblock Moderate deviations for quantile regression processes.
\newblock {\em Communications in Statistics-Theory and Methods},
  48(12):2879--2892.

\bibitem[Maynard et~al., 2023]{maynard2023inference}
Maynard, A., Shimotsu, K., and Kuriyama, N. (2023).
\newblock Inference in predictive quantile regressions.
\newblock {\em arXiv preprint arXiv:2306.00296}.

\bibitem[Mikusheva, 2007]{mikusheva2007uniform}
Mikusheva, A. (2007).
\newblock Uniform inference in autoregressive models.
\newblock {\em Econometrica}, 75(5):1411--1452.

\bibitem[Mikusheva, 2012]{mikusheva2012one}
Mikusheva, A. (2012).
\newblock One-dimensional inference in autoregressive models with the potential
  presence of a unit root.
\newblock {\em Econometrica}, 80(1):173--212.

\bibitem[Neocleous and Portnoy, 2008]{neocleous2008monotonicity}
Neocleous, T. and Portnoy, S. (2008).
\newblock On monotonicity of regression quantile functions.
\newblock {\em Statistics \& probability letters}, 78(10):1226--1229.

\bibitem[Oh et~al., 2018]{oh2018mildly}
Oh, H., Lee, S., and Chan, N.~H. (2018).
\newblock Mildly explosive autoregression with mixing innovations.
\newblock {\em Journal of the Korean Statistical Society}, 47(1):41--53.

\bibitem[Phillips et~al., 2001]{phillips2001estimate}
Phillips, P.~C., Moon, H.~R., and Xiao, Z. (2001).
\newblock How to estimate autoregressive roots near unity.
\newblock {\em Econometric Theory}, 17(1):29--69.

\bibitem[Phillips, 1987a]{phillips1987time}
Phillips, P. C.~B. (1987a).
\newblock Time series regression with a unit root.
\newblock {\em Econometrica: Journal of the Econometric Society}, pages
  277--301.

\bibitem[Phillips, 1987b]{phillips1987towards}
Phillips, P. C.~B. (1987b).
\newblock Towards a unified asymptotic theory for autoregression.
\newblock {\em Biometrika}, 74(3):535--547.

\bibitem[Phillips, 2014]{phillips2014confidence}
Phillips, P. C.~B. (2014).
\newblock On confidence intervals for autoregressive roots and predictive
  regression.
\newblock {\em Econometrica}, 82(3):1177--1195.

\bibitem[Phillips and Magdalinos, 2007]{Phillips2007limit}
Phillips, P. C.~B. and Magdalinos, T. (2007).
\newblock Limit theory for moderate deviations from a unit root.
\newblock {\em Journal of Econometrics}, 136(1):115--130.

\bibitem[Phillips and Magdalinos, 2009]{PM2009econometric}
Phillips, P. C.~B. and Magdalinos, T. (2009).
\newblock Econometric inference in the vicinity of unity.
\newblock {\em Singapore Management University, CoFie Working Paper}, 7.

\bibitem[Pollard, 1991]{pollard1991asymptotics}
Pollard, D. (1991).
\newblock Asymptotics for least absolute deviation regression estimators.
\newblock {\em Econometric Theory}, 7(2):186--199.

\bibitem[Pro{\"\i}a, 2020]{proia2020moderate}
Pro{\"\i}a, F. (2020).
\newblock Moderate deviations in a class of stable but nearly unstable
  processes.
\newblock {\em Journal of Statistical Planning and Inference}, 208:66--81.

\bibitem[Rao, 1978]{rao1978asymptotic}
Rao, M. (1978).
\newblock Asymptotic distribution of an estimator of the boundary parameter of
  an unstable process.
\newblock {\em The Annals of Statistics}, pages 185--190.

\bibitem[Saxena and Alam, 1982]{saxena1982estimation}
Saxena, K.~L. and Alam, K. (1982).
\newblock Estimation of the non-centrality parameter of a chi squared
  distribution.
\newblock {\em The Annals of Statistics}, pages 1012--1016.

\bibitem[Vogelsang, 1998]{vogelsang1998trend}
Vogelsang, T.~J. (1998).
\newblock Trend function hypothesis testing in the presence of serial
  correlation.
\newblock {\em Econometrica}, pages 123--148.

\bibitem[Wang et~al., 2022]{wang2022asymptotics}
Wang, X., Tang, X., and Song, Y. (2022).
\newblock Asymptotics of m-estimators for moderate deviations from a unit root
  model with possibly infinite variance.
\newblock {\em Communications in Statistics-Theory and Methods}, pages 1--18.

\bibitem[Werker and Zhou, 2022]{werker2022semiparametric}
Werker, B.~J. and Zhou, B. (2022).
\newblock Semiparametric testing with highly persistent predictors.
\newblock {\em Journal of Econometrics}, 227(2):347--370.

\bibitem[White, 1958]{white1958limiting}
White, J.~S. (1958).
\newblock The limiting distribution of the serial correlation coefficient in
  the explosive case.
\newblock {\em The Annals of Mathematical Statistics}, pages 1188--1197.

\bibitem[Xu and Pang, 2018]{xu2018limit}
Xu, C. and Pang, T. (2018).
\newblock Limit theory for moderate deviations from a unit root with a break in
  variance.
\newblock {\em Communications in Statistics-Theory and Methods},
  47(24):6125--6143.

\bibitem[Yabe, 2012]{yabe2012limiting}
Yabe, R. (2012).
\newblock Limiting distribution of the score statistic under moderate deviation
  from a unit root in ma (1).
\newblock {\em Journal of Time Series Analysis}, 33(4):533--541.

\bibitem[Yu and Kejriwal, 2021]{yu2021inference}
Yu, X. and Kejriwal, M. (2021).
\newblock Inference in mildly explosive autoregressions under unconditional
  heteroskedasticity.
\newblock {\em Available at SSRN 4081663}.

\bibitem[Zhou and Lin, 2015]{zhou2015quantile}
Zhou, Z. and Lin, Z. (2015).
\newblock Quantile inference for moderate deviations from a unit root model
  with infinite variance.
\newblock {\em Journal of the Korean Statistical Society}, 44(2):280--294.

\end{thebibliography}

\end{document}